\documentclass{ws-ijfcs}

\usepackage{enumerate}
\usepackage{xcolor}
\usepackage[hyperfootnotes=false,colorlinks,urlcolor=darkblue,linkcolor=darkblue,citecolor=darkblue,]{hyperref}

\definecolor{darkblue}{RGB}{0,0,180}

\newcommand{\te}{\theta}

\newcommand{\Pref}{\ensuremath{\mathrm{Pref}}}
\newcommand{\Suff}{\ensuremath{\mathrm{Suff}}}
\newcommand{\lcm}{\ensuremath{\mathrm{lcm}}}

\newtheorem{problem}{Problem}
\newtheorem{claim}{Claim}

\begin{document}

\markboth{L. Kari, B. Masson, S. Seki}
{Properties of Pseudo-Primitive Words and their Applications}

%
\catchline{}{}{}{}{}
%

\title{PROPERTIES OF PSEUDO-PRIMITIVE WORDS\\AND THEIR APPLICATIONS}
\author{LILA KARI, BENO\^{I}T MASSON, SHINNOSUKE SEKI}
\address{Department of Computer Science, The University of Western Ontario, \\ 
London, Ontario, N6A 5B7, Canada \\
\email{\{lila, benoit, sseki\}@csd.uwo.ca}
}

\maketitle

\begin{history}
\received{(Day Month Year)}
\accepted{(Day Month Year)}
\comby{(xxxxxxxxxx)}
\end{history}

\begin{abstract}
	A pseudo-primitive word with respect to an antimorphic involution $\te$ is a word which cannot be written as a catenation of occurrences of a strictly shorter word $t$ and $\te(t)$. 
	Properties of pseudo-primitive words are investigated in this paper. 
	These properties link pseudo-primitive words with essential notions in combinatorics on words such as primitive words, (pseudo)-palindromes, and (pseudo)-commutativity.
	Their applications include an improved solution to the  extended Lyndon-Sch\"{u}tzenberger equation $u_1u_2 \cdots u_\ell = v_1 \cdots v_n w_1 \cdots w_m$, where $u_1, \ldots, u_\ell \in \{u, \theta(u)\}$, $v_1, \ldots, v_n \in \{v, \theta(v)\}$, and $w_1, \ldots, w_m \in \{w, \theta(w)\}$ for some words $u, v, w$, integers $\ell, n, m \ge 2$, and an antimorphic involution $\te$. 
	We prove that for $\ell \ge 4$, $n, m \ge 3$, this equation implies that $u, v, w$ can be expressed in terms of a common word $t$ and its image $\theta(t)$. 
	Moreover, several cases  of this equation where $\ell = 3$ are examined. 
\end{abstract}


\ccode{2000 Mathematics Subject Classification: 68Q70, 68R15}


	\section{Introduction}

For elements $u, v, w$ in a free group, the equation of the form $u^\ell = v^n w^m$ $(\ell, n, m \ge 2)$ is known as the {\it Lyndon-Sch\"{u}tzenberger equation} (LS equation for short). 
Lyndon and Sch\"{u}tzenberger~\cite{LySc62} investigated the question of finding all possible solutions for this equation in a free group, and proved that if the equation holds, then $u$, $v$, and $w$ are all powers of a common element. 
This equation can be also considered on the semigroup of all finite words over a fixed alphabet $\Sigma$, and an analogous result holds. 

\begin{theorem}[see, e.g., \cite{HaNo04,LySc62,Man02}]\label{thm:original}
	For words $u, v, w \in \Sigma^+$ and integers $\ell, n, m \ge 2$, the equation $u^\ell = v^n w^m$ implies that $u, v, w$ are powers of a common word. 
\end{theorem}

The Lyndon-Sch\"{u}tzenberger equation has been generalized in several ways; e.g., the equation of the form $x^k = z_1^{k_1} z_2^{k_2}
 \cdots z_n^{k_n}$ was investigated by Harju and Nowotka~\cite{HaNo05} and its special cases in~\cite{ApDj68,Len65}. 
Czeizler et al.~\cite{CCKS09} have recently proposed another extension, which was originally motivated by the information encoded as
 DNA strands for DNA computing. 
In this framework, a DNA strand is modeled by a word $w$ and encodes the same information as its Watson-Crick complement. 
In formal language theory, the Watson-Crick complementarity of DNA strands is modeled by an antimorphic involution
 $\te$~\cite{KKT02,PRY01}, i.e., a function $\te$ on an alphabet $\Sigma^*$ that is 
{\em (a)} antimorphic, $\te(x y) = \te(y) \te(x)$, $\forall x, y \in \Sigma^*$, and {\em (b)} involution, $\te^2 = id$, the identity. 
Thus, we can model the property  whereby a DNA single strand binds to and 
 is completely equivalent to its Watson-Crick complement, by  considering 
a word $u$ and its image $\te (u)$ equivalent, for a given antimorphic involution $\te$. 

For words $u$, $v$, $w$, integers $\ell, n, m \ge 2$, and an antimorphic involution $\te$, an extended  Lyndon-Sch\"{u}tzenberger
 equation (ExLS equation) is  of the form 
\begin{equation}\label{eq:exLS}
	u_1u_2 \cdots u_\ell = v_1 \cdots v_n w_1 \cdots w_m,
\end{equation}
with $u_1, \ldots, u_\ell \in \{u, \te(u)\}$, $v_1, \ldots, v_n \in \{v, \te(v)\}$, and $w_1, \ldots, w_m \in \{w, \te(w)\}$. 
The question arises as to whether an equation of this form implies the existence of a word $t$ such that $u, v, w \in \{t, \te(t)\}^+$. 
A given triple $(\ell, n, m)$ of integers is said to {\it impose pseudo-periodicity, with respect to $\te$, on $u, v, w$}, or simply, to {\it impose $\te$-periodicity on $u, v, w$} if (\ref{eq:exLS}) implies $u, v, w \in \{t, \te(t)\}^+$ for some word $t$. 
Furthermore, we say that the triple $(\ell, n, m)$ {\it imposes $\te$-periodicity} if it imposes $\te$-periodicity on all $u, v, w$. 
The known results on ExLS equations \cite{CCKS09} are summarized in Table~\ref{tbl:exLS_summary}.

\begin{table}[h]
\label{tbl:exLS_summary}
\tbl{Summary of the known results regarding the extended Lyndon-Sch\"{u}tzenberger equation.}
{\begin{tabular}{r@{\hspace{8mm}}r@{\hspace{8mm}}r@{\hspace{10mm}}c}
\toprule
$l$ & $n$ & $m$ & $\theta$-periodicity \\
\colrule
$\ge 5$ & $\ge 3$ & $\ge 3$ & YES \\
\colrule
$3$ or $4$ & $\ge 3$ & $\ge 3$ & ? \\
$2$ & $\ge 2$ & $\ge 2$ & ? \\
\colrule
$\ge 3$ & $2$ & $\ge 2$ & NO \\
\botrule
\end{tabular}}
\end{table}

This paper is a step towards  solving  the unsettled cases of  Table~\ref{tbl:exLS_summary}, by using the  following strategy.
Concise proofs exist in the literature for Theorem~\ref{thm:original}, that make use of fundamental properties such as: 
\begin{enumerate}[(i)]
\item	The periodicity theorem of Fine and Wilf (FW theorem), \label{item:FW} 
\item	The fact that a primitive word cannot be a proper infix of its square, and \label{item:overlap}
\item	The fact that the  class of primitive words is closed under cyclic permutation. \label{item:cycle}
\end{enumerate}
(For details of each, see~\cite{ChKa97}.) 
In contrast, the proof given in~\cite{CCKS09} for the result about ExLS equations, stating that $(\ge 5, \ge 3, \ge 3)$ imposes
 $\te$-periodicity, involves techniques designed  for this specific purpose only.
Should Properties (\ref{item:FW}), (\ref{item:overlap}), (\ref{item:cycle}) be generalized so as to take into account the
 informational equivalence between a word $u$ and $\te(u)$,  they could possibly form a basis for a concise proof of the
 solutions to the ExLS equation. The approach we use in this paper is thus to seek analog properties for this extended case,
 and use the results we obtain to approach the unsettled cases in   Table~\ref{tbl:exLS_summary}. 

Czeizler, Kari, and Seki generalized Property~(\ref{item:FW}) in \cite{CKS08}. 
There,  first  the notion of primitive words was extended to that of pseudo-primitive words with respect to a given antimorphic
 involution $\te$ (or simply, $\te$-primitive words).
A word $u$ is said to be {\it $\te$-primitive} if there does not exist another word $t$ such that $u \in t\{t, \te(t)\}^+$.  
For example, if $\te$ is the mirror image over  $\{a, b\}^*$  (the identity function on $\{a, b\}$ extended to an 
 antimorphism on $\{a, b\}^*$), $aabb$ is $\te$-primitive, while $abba$ is not because it can be written as $ab \te(ab)$. 
Based on the $\te$-primitivity of words, Property~(\ref{item:FW}) was generalized as  follows: 
``For words $u, v$, if a word in $u\{u, \te(u)\}^*$ and a word in $v\{v, \te(v)\}^*$ share a long enough prefix (for details, see Theorems~\ref{thm:exFWlcm} and \ref{thm:exFWgcd}), then $u, v \in t\{t, \te(t)\}^*$ for some $\te$-primitive word $t$.'' 

In contrast, little is known about Properties (\ref{item:overlap}) and (\ref{item:cycle}) except that they cannot be 
 generalized as suggested in the previous example: non-trivial overlaps between two words in $\{t, \te(t)\}^+$ are possible,
 and cyclic permutations do  not in general preserve the $\te$-primitivity of words. 
As a preliminary step towards an extension of Property (\ref{item:overlap}), Czeizler et al.~examined the non-trivial overlap of the form $v_1 \cdots v_m x = y v_{m+1} \cdots v_{2m}$, where $m \ge 1$, $v_i$ is either $v$ or $\te(v)$ for some $\te$-primitive word $v$ $(1 \le i \le 2m)$, and both $x$ and $y$ are properly shorter than $v$ \cite{CCKS09}. 
Some of the results obtained there will be employed in this paper. 

One purpose of this paper is to explore further the extendability of Properties (\ref{item:overlap}) and (\ref{item:cycle}). 
The main result here is Theorem~\ref{thm:overlap3}, which states that for a $\te$-primitive word $x$, neither $x\te(x)$ nor $\te(x)x$ can be a proper infix of a word $x_1x_2x_3$, where $x_1, x_2, x_3 \in \{x, \te(x)\}$. 
Based on this result, we consider two problems: For a $\te$-primitive word $x$, 
(1) does $v, yvz \in \{x, \te(x)\}^+$ imply that $y$ and $z$ are in $\{x, \te(x)\}^*$?; and 
(2) if the catenation of words $u, v$ is in $\{x, \te(x)\}^+$, under what conditions does $u, v \in \{x, \te(x)\}^*$ hold? 
In particular, our investigation into the second problem will reveal close relationships between primitive words, $\te$-primitive words, and $\te$-palindromes (fixed points of $\te$). These relationships further present several cyclic permutations under which the $\te$-primitivity of words is preserved. 

The results thus obtained enable us to prove  that the triple $(4, \ge 3, \ge 3)$ imposes $\te$-periodicity (Theorem~\ref{thm:exLS4}) in a much simpler manner than the proof in \cite{CCKS09} for $(\ge 5, \ge 3, \ge 3)$. 
Even for $(3, n, m)$ ExLS equations, these results give some insight and narrow down the open cases of ExLS equations. 

The paper is organized as follows: in the next section, we provide required notions and notation. 
Section~\ref{sec:combinatorial} begins with the proof of some basic properties of $\te$-primitive words, and then proves some consequences of overlaps between $\te$-primitive words of a similar flavour with  Properties~(\ref{item:overlap}) and (\ref{item:cycle}) (e.g., Theorem~\ref{thm:overlap3}, Corollary~\ref{cor:conj_pal}).
These tools are used in Section~\ref{sec:ExLS}, where we prove that the $(4, \ge 3, \ge 3)$ ExLS equation has only $\te$-periodic solutions (Theorem~\ref{thm:exLS4}), and study particular cases of $(3, n, m)$ ExLS equations.

	\section{Preliminaries}\label{sec:pre}

An alphabet is a finite and non-empty set of symbols. 
In the sequel, we shall use a fixed non-singleton alphabet $\Sigma$. 
The set of all words over $\Sigma$ is denoted by $\Sigma^*$, which includes the empty word $\lambda$, and let $\Sigma^+ = \Sigma^* \setminus \{\lambda\}$. 
The length of a word $w \in \Sigma^*$ is denoted by $|w|$. 
A word $v$ is an {\it infix} (resp. {\it prefix}, {\it suffix}) of a word $w$ if $w = xvy$ (resp. $w = vy$, $w = xv$) for some $x, y \in \Sigma^*$; in any case, if $w \neq v$, then the infix (prefix, suffix) is said to be {\it proper}. 
For a word $w$, denote by $\Pref(w)$ the set of prefixes of $w$ and by $\Suff(w)$ the set of its suffixes.

A language $L$ is a subset of $\Sigma^*$. 
For a non-negative integer $n \ge 0$, we write $L^n$ for the language consisting of all words of the form $w_1 \cdots w_n$ such that each $w_i$ is in $L$. 
We also write $L^{\ge n}$ for $L^n \cup L^{n+1} \cup L^{n+2} \cup \cdots$. 
Analogously, we can define $L^{\le n} = L^0 \cup L^1 \cup \cdots \cup L^n$. 
For $L^{\ge 0}$ and $L^{\ge 1}$, we employ the traditional notation $L^*$ and $L^+$. 

A mapping $\te: \Sigma^* \to \Sigma^*$ is called an {\it antimorphic involution} of $\Sigma^*$ if $\te(xy) = \te(y)\te(x)$ for any $x, y \in \Sigma^*$ (antimorphism), and $\te^2$ is equal to the identity (involution). 
Throughout this paper, $\te$ denotes an antimorphic involution. 
The {\it mirror image}, which maps a word to its reverse, is a typical example of antimorphic involution. 
A word $w \in \Sigma^*$ is called a {\it $\te$-palindrome} if $w = \te(w)$.
A word which is a $\te$-palindrome with respect to a given but unspecified antimorphic involution $\te$ is also called a {\it pseudo-palindrome} \cite{dLDL06a}. 

A non-empty word $w \in \Sigma^+$ is said to be {\it primitive} if $w = v^n$ implies $n = 1$ for any word $v \in \Sigma^+$. 
It is known that any non-empty word $w \in \Sigma^+$ can be written as a power of a unique primitive word, which is called the {\it primitive root} of $w$, and denoted by $\rho(w)$. 
Two words which {\it commute} share a primitive root, that is, $uv = vu$ implies $\rho(u) = \rho(v)$ (see~\cite{ChKa97}). 
In literature, it is said that $uv = vu$ causes a {\it defect effect} (for details of defect effects and defect theorems, see~\cite{ChKa97,Man02}). 
The LS equation also causes defect effect, since $u^\ell = v^n w^m$ with $\ell, n, m \ge 2$ implies $\rho(u) = \rho(v) = \rho(w)$ (Theorem~\ref{thm:original}). 
The following results describe other relations causing a defect effect.

\begin{lemma}[\cite{CKS08}]\label{lem:pq-qp}
	Let $u \in \Sigma^+$ such that $u = p q$ for some $\te$-palindromes $p, q \in \Sigma^+$. If $q \in \Pref(u)$ and $|q| \geq |p|$, then $\rho(p) = \rho(q) = \rho(u)$.
\end{lemma}

\begin{theorem}[\cite{ChKa97}]\label{th:uv-expr}
	Let $u, v \in \Sigma^+$. 
	If there exist $\alpha(u, v) \in u\{u, v\}^*$ and $\beta(u, v) \in v\{u, v\}^*$ which share a prefix of length at least $|u| + |v|$, then $\rho(u) = \rho(v)$. 
\end{theorem}

The notion of primitive word was generalized into that of pseudo-primitive word by Czeizler, Kari, and Seki~\cite{CKS08}. 
For an antimorphic involution $\te$, a non-empty word $w \in \Sigma^+$ is said to be {\it pseudo-primitive with respect to $\te$}, or simply {\it $\te$-primitive}, if $w \in \{v, \te(v)\}^n$ implies $n = 1$ for any word $v \in \Sigma^+$. 
In \cite{CKS08} it was proved that for any non-empty word $w \in \Sigma^+$, there exists a unique $\te$-primitive word $t$ satisfying $w \in t\{t, \te(t)\}^*$. 
Such a word $t$ is called the {\it $\te$-primitive root} of $w$. 
The next lemma describes a property of the $\te$-primitive root of a $\te$-palindrome of even length. 

\begin{lemma}\label{lem:pali_even}
	Let $x \in \Sigma^+$ be a $\te$-primitive word and $p$ be a $\te$-palindrome of even length. 
	If $p = x_1 x_2 \cdots x_m$ for some $m \ge 1$ and $x_1, \ldots, x_m \in \{x, \te(x)\}$, then $m$ has to be even. 
\end{lemma}
\begin{proof}
	Suppose that the equality held for some odd $m$. 
	Then $x$ must be of even length because $|p|$ is even. 
	Hence $x_{(m-1)/2}$ becomes a $\te$-palindrome. 
	Thus $x = y\te(y)$ for some $y \in \Sigma^+$. 
	However, this contradicts the $\te$-primitivity of $x$. 
\end{proof}

The {\it theorem of Fine and Wilf} (FW theorem) is one of the fundamental results on periodicity \cite{FiWi65}. 
It states that for two words $u, v \in \Sigma^+$, if a power of $u$ and a power of $v$ share a prefix of length at least $|u|+|v|-\gcd(|u|, |v|)$, then $\rho(u) = \rho(v)$, where $\gcd(\cdot, \cdot)$ denotes the greatest common divisor of two arguments (for its proof, see, e.g.,~\cite{ChKa97}). 
This theorem has been generalized in~\cite{CKS08}, by taking into account the equivalence between a word and its image under $\te$, in the following two forms. 

\begin{theorem}[\cite{CKS08}]\label{thm:exFWlcm}
	Let $u, v \in \Sigma^+$. 
	If a word in $\{u, \te(u)\}^*$ and a word in $\{v, \te(v)\}^*$ share a prefix of length at least $\lcm(|u|, |v|)$, then $u, v \in \{t, \te(t)\}^+$ for some $\te$-primitive word $t \in \Sigma^+$, where $\lcm(\cdot, \cdot)$ denotes the least common multiple of two arguments. 
\end{theorem}

\begin{theorem}[\cite{CKS08}]\label{thm:exFWgcd}
	Let $u, v \in \Sigma^+$ with $|u| \ge |v|$. 
	If a word in $\{u, \te(u)\}^*$ and a word in $\{v, \te(v)\}^*$ share a prefix of length at least $2|u|+|v|-\gcd(|u|, |v|)$, then $u, v \in \{t, \te(t)\}^+$ for some $\te$-primitive word $t \in \Sigma^+$. 
\end{theorem}

In a way, we can say that these theorems describe relations causing a {\it weak defect effect} because they all imply that $u, v \in \{t, \te(t)\}^+$ for some $\te$-primitive word $t \in \Sigma^+$, which is strictly weaker than the usual defect effect $\rho(u) = \rho(v)$~\cite{CKS08}.
Various relations causing such a weak defect effect were presented in~\cite{CKS08}. 

Besides, the commutativity $xy = yx$ was extended to the $\te$-commutativity $xy = \te(y)x$ in~\cite{KaMa08}. 
This is a special case of $xy = zx$, whose solutions are given as $x = r(tr)^i$, $y = (tr)^j$, and $z = (rt)^j$ for some $i \ge 0$, $j \ge 1$, and $r, t \in \Sigma^*$ such that $rt$ is primitive (see, e.g.,~\cite{ChKa97}).
The next proposition immediately follows from this; note that the $\te$-commutativity equation guarantees that both $r, t$ are $\te$-palindromes. 

\begin{proposition}[\cite{KaMa08}]\label{prop:th-commute}
	For $x, y \in \Sigma^+$, the solutions of $xy = \te(y)x$ are given by $x = r(tr)^i$ and $y = (tr)^j$ for some $i \ge 0$, $j \ge 1$, and $\te$-palindromes $r, t$ such that $rt$ is primitive. 
\end{proposition}

Although this equation does not cause even a weak defect effect, one encounters it often when considering word equations which involve $\te$. 
Note that for words $u, v \in \Sigma^*$, it was proved in~\cite{CKS08} that the system $uv = \te(uv)$ and $vu = \te(vu)$ causes a weak defect effect: $u, v \in \{t, \te(t)\}^*$ for some $t \in \Sigma^+$. 
Thus for words $x, y, z$ satisfying $xy = zx$, if both $y$ and $z$ are $\te$-palindromes, then the representation of solutions of $xy = zx$ implies $tr = \te(tr)$ and $rt = \te(rt)$. 
Hence the next result holds. 

\begin{proposition}[\cite{CCKS09}]\label{prop:pali_conjugate}
	For a word $x \in \Sigma^+$ and two $\te$-palindromes $y, z \in \Sigma^+$, the equation $xy = zx$ implies that $x, y, z \in \{t, \te(t)\}^*$ for some $t \in \Sigma^+$. 
\end{proposition}

	\section{Properties of Pseudo-Primitive Words}\label{sec:combinatorial}

The primitivity of words is one of the most essential notions in combinatorics on words. 
The past few decades saw a considerable number of studies on this topic (see e.g.,~\cite{ChKa97,Lot83,Shy01}). 
In contrast, research on the pseudo-primitivity of words has just been initiated in~\cite{CCKS09,CKS08}. 
For instance, although the class of pseudo-primitive words was proved to be properly included in that of primitive words~\cite{CKS08}, nothing else is known about the relation between these two classes. 
The purpose of this section is to prove various properties of pseudo-primitive words. 

Throughout this section, $\te$ is assumed to be a given antimorphic involution. 
We begin this section with a simple extension of a known result on the primitive root (Lemma~\ref{lem:rt_rootshare}) to the $\te$-primitive root (Lemma~\ref{lem:rt_th-rootshare}). 

\begin{lemma}[e.g.,~\cite{Shy01}]\label{lem:rt_rootshare}
	For words $u, v \in \Sigma^+$ and a primitive word $w \in \Sigma^+$, the following properties hold: 
	\begin{arabiclist}
	\item $u^n \in w^+$ implies $u \in w^+$; 
	\item $uv, u \in w^+$ or $uv, v \in w^+$ implies $u, v \in w^+$. 
	\end{arabiclist}
\end{lemma}

\begin{lemma}\label{lem:rt_th-rootshare}
	For words $u, v \in \Sigma^+$ and a $\te$-primitive word $x \in \Sigma^+$, the following properties hold:
	\begin{arabiclist}
	\item for some $n \ge 1$, $u^n \in \{x, \te(x)\}^+$ implies $u \in \{x, \te(x)\}^+$; 
	\item $uv, u \in \{x, \te(x)\}^+$, or $uv, v \in \{x, \te(x)\}^+$ implies $u, v \in \{x, \te(x)\}^+$; 
	\item $\te(u)v, u \in \{x, \te(x)\}^+$, or $u\te(v), v \in \{x, \te(x)\}^+$ implies $u, v \in \{x, \te(x)\}^+$. 
	\end{arabiclist}
\end{lemma}
\begin{proof}
	The first property follows from Theorem~\ref{thm:exFWlcm}, while the others are immediately proved by comparing the length of words.
\end{proof}

As mentioned in the introduction, if a word $w$ is primitive, then the equation $w^2 = ywz$ implies either $y = \lambda$ or $z = \lambda$. 
Since a $\te$-primitive word is primitive, this applies to $\te$-primitive words, too; a $\te$-primitive word $x$ cannot be a proper infix of $x^2$. 
However, due to the informational equivalence between $x$ and $\te(x)$, we should consider equations like $x^2 = y\te(x)z$ as well, and in fact this equation can hold with non-empty $y$ and $z$. 
Nevertheless, we can state an analogous theorem based on the next lemma.

\begin{lemma}[\cite{CKS08}]\label{lem:overlap2}
	Let $x \in \Sigma^+$ be a $\te$-primitive word, and $x_1$, $x_2$, $x_3$, $x_4 \in \{x, \te(x)\}$. 
	If $x_1 x_2 y = z x_3 x_4$ for some non-empty words $y, z \in \Sigma^+$ with $|y|, |z| < |x|$, then $x_2 \neq x_3$. 
\end{lemma}

\begin{theorem}\label{thm:overlap3}
	For a $\te$-primitive word $x \in \Sigma^+$, neither $x\te(x)$ nor $\te(x)x$ can be a proper infix of a word in $\{x, \te(x)\}^3$.
\end{theorem}
\begin{proof}
	Let $x_1, x_2, x_3 \in \{x, \te(x)\}$ and suppose that $x\te(x)$ is a proper infix of $x_1x_2x_3$. 
	That is to say, there exist words $y, z, y', z' \in \Sigma^+$, $0 < |y|, |z|, |y'|, |z'| < |x|$ such that $zx\te(x) = x_1x_2y$ and $x\te(x)y' = z'x_2x_3$. 
	By Lemma~\ref{lem:overlap2}, the first equation implies that $x_2 \ne x$ and the second that $x_2 \ne \te(x)$, this is in contradiction with $x_2 \in \{x, \te(x)\}$. 
	We prove similarly that $\te(x)x$ cannot be a proper infix of $x_1x_2x_3$.
\end{proof}

This theorem will lead us to two propositions (Propositions~\ref{prop:pali_split} and~\ref{prop:clean_split}), as well as to several other results. 
The main usage of these propositions in this paper is the following ``splitting strategy,'' which shall prove useful in solving ExLS equations in Section~\ref{sec:ExLS}.  
Given ``complicated'' words in $\{x, \te(x)\}^+$ for a $\te$-primitive word $x$, these propositions make it possible to split such words into ``simple'' component words which are still in $\{x, \te(x)\}^+$.  
Then, Lemmas~\ref{lem:rt_rootshare} and~\ref{lem:rt_th-rootshare} are often applicable to subdivide these simple components into smaller units in $\{x, \te(x)\}^+$.

Recall that a primitive word cannot be a proper infix of its square. 
It is hence evident that for a primitive word $w$, if a word $u$ in $w^+$ contains $w$ as its infix like $u = ywz$ for some $y, z \in \Sigma^*$, then $y, z \in w^*$. 
For such $w$, more generally, $v, yvz \in w^+$ implies $y, z \in w^*$. 
This raises a naturally extended question of whether for a $\te$-primitive word $x$, if $v, yvz \in \{x, \te(x)\}^+$, then $y, z \in \{x, \te(x)\}^*$ holds or not. 
Although this is not always the case, we provide some positive cases based on the following lemma, which is a natural consequence of Theorem~\ref{thm:overlap3}. 

\begin{lemma}\label{lem:xtex_infix1}
	Let $x$ be a $\te$-primitive word, and $v \in \Sigma^+$. 
	For $y, z \in \Sigma^*$, either $yx\te(x)z \in \{x, \te(x)\}^*$ or $y\te(x)xz  \in \{x, \te(x)\}^*$ implies $y, z \in \{x, \te(x)\}^*$. 
\end{lemma}
\begin{proof}
	We prove that $yx\te(x)z \in \{x, \te(x)\}^*$ implies $y, z \in \{x, \te(x)\}^*$. 
	Let $yx\te(x)z = x_1 \cdots x_n$ for some $n \ge 2$ and $x_1, \ldots, x_n \in \{x, \te(x)\}$. 
	In light of Theorem \ref{thm:overlap3}, there must exist such $i$ that $y = x_1 \cdots x_{i-1}$, $x\te(x) = x_ix_{i+1}$, and $z = x_{i+2} \cdots x_n$. 
\end{proof}

\begin{lemma}\label{lem:xtex_infix2}
	Let $x$ be a $\te$-primitive word, and $v \in \Sigma^+$. 
	If $v, yvz \in \{x, \te(x)\}^*$ for some $y, z \in \Sigma^*$ and either $x\te(x)$ or $\te(x)x$ is an infix of $v$, then $y, z \in \{x, \te(x)\}^*$. 
\end{lemma}
\begin{proof}
	Here we consider only the case when $x\te(x)$ is an infix of $v$. 
	Due to Lemma \ref{lem:xtex_infix1}, we can let $v = x'x\te(x)x''$ for some $x', x'' \in \{x, \te(x)\}^*$. 
	Thus, $yvz = yx' x\te(x) x''z \in \{x, \te(x)\}^{\ge 2}$. 
	From this, the same lemma derives $yx', x''z \in \{x, \te(x)\}^*$. 
	Based on Lemma \ref{lem:rt_th-rootshare}, we obtain $y, z \in \{x, \te(x)\}^*$. 
\end{proof}

Lemma~\ref{lem:xtex_infix2} is a generalization of Lemma~\ref{lem:xtex_infix1}, and makes it possible to prove the following two propositions. 

\begin{proposition}
	Let $x$ be a $\te$-primitive word, and $v \in \Sigma^+$. 
	If $v, yvz \in \{x, \te(x)\}^{\ge 2}$ for some $y, z \in \Sigma^*$ and $v$ is primitive, then $y, z \in \{x, \te(x)\}^*$. 
\end{proposition}
\begin{proof}
	Let $v = x_1 \cdots x_m$ for some $m \ge 2$ and $x_1, \ldots, x_m \in \{x, \te(x)\}$. 
	Since $v$ is primitive, there exists $1 \le i \le m$ such that $x_ix_{i+1} \in \{x\te(x), \te(x)x\}$. 
	Now we can employ Lemma \ref{lem:xtex_infix2} to get this result. 
\end{proof}

\begin{proposition}\label{prop:pali_split}
	Let $x$ be a $\te$-primitive word, and $v \in \Sigma^+$. 
	If $v, yvz \in \{x, \te(x)\}^+$ for some $y, z \in \Sigma^*$ and $v$ is a non-empty $\te$-palindrome, then $y, z \in \{x, \te(x)\}^*$. 
\end{proposition}
\begin{proof}
	Let $v = x_1 \cdots x_n$ for some $n \ge 1$ and $x_1, \ldots, x_n \in \{x, \te(x)\}$. 
	If $n$ is odd, then $v = \te(v)$ implies $x_{(n+1)/2} = \te(x_{(n+1)/2})$ and this means $x = \te(x)$. 
	Thus we have $v, yvz \in x^+$, and hence $y, z \in x^*$. 
	If $n$ is even, then $x_{n/2}x_{n/2+1} \in \{x\te(x), \te(x)x\}$ so that $y, z \in \{x, \te(x)\}^*$ due to Lemma~\ref{lem:xtex_infix2}. 
\end{proof}

>From now on, we address the following question: 
``for a $\te$-primitive word $x$ and two words $u, v \in \Sigma^*$ such that $uv \in \{x, \te(x)\}^+$, under what conditions on $u, v$, we can say $u, v \in \{x, \te(x)\}^*$?''. 
Here we provide several such conditions. 
Among them is Proposition~\ref{prop:clean_split}, which serves for the splitting strategy. 
As its corollary, we will obtain relationships between primitive words and $\te$-primitive words (Corollaries~\ref{cor:prime-te-prime} and \ref{cor:pq2_te-prime}). 

\begin{proposition}\label{prop:pref_suff_split}
	Let $x$ be a $\te$-primitive word, $u \in \Suff(\{x, \te(x)\}^+)$, and $v \in \Pref(\{x, \te(x)\}^+)$. 
	If $uv = x_1 \cdots x_m$ for some integer $m \ge 2$ and $x_1, \ldots, x_m \in \{x, \te(x)\}$, then either $u, v \in \{x, \te(x)\}^+$ or $x_1 = \cdots = x_m$. 
\end{proposition}
\begin{proof}
	Let us prove that when $u, v \not\in \{x, \te(x)\}^+$, $x_1 = \cdots = x_m$ must hold. 
	Let $u = z_s' x_{i-1}' \cdots x_1'$ for some $i \ge 1$, $x_i', \ldots, x_1' \in \{x, \te(x)\}$, and some non-empty words $z_p', z_s' \in \Sigma^+$ such that $z_p'z_s' = x'_i$. 
	We can also let $v = x_1'' \cdots x_{j-1}'' z_p''$ for some $j \ge 1$, $x_1'', \ldots, x_j'' \in \{x, \te(x)\}$, and $z_p'', z_s'' \in \Sigma^+$ such that $z_p''z_s'' = x_j$.
	Now we have $x_i' \cdots x_1' x_1'' \cdots x_j'' = z_p' uv z_s'' = z_p' x_1 \cdots x_m z_s''$. 
	Since $0 < |z_p'| < |x|$, Theorem~\ref{thm:overlap3} implies $x_1 = \cdots = x_m$. 
\end{proof}

\begin{corollary}
	Let $x$ be a $\te$-primitive word, and $u \in \Suff(\{x, \te(x)\}^+)$, $v \in \Pref(\{x, \te(x)\}^+)$. 
	If $uv$ is in $\{x, \te(x)\}^{\ge 2}$ and primitive, then $u, v \in \{x, \te(x)\}^+$. 
\end{corollary}

Proposition~\ref{prop:pref_suff_split} gives the following two propositions which play an important role in investigating the ExLS equation. 

\begin{proposition}\label{prop:conjugacy}
	Let $x$ be a $\te$-primitive word, and $u, v \in \Sigma^+$. 
	If $uv, vu \in \{x, \te(x)\}^n$ for some $n \ge 2$, then one of the following statements holds:
	\begin{arabiclist}
	\item 	$u, v \in \{x, \te(x)\}^+$; 
	\item	$uv = x^n$ and $vu = \te(x)^n$; 
	\item	$uv = \te(x)^n$ and $vu = x^n$. 
	\end{arabiclist}
\end{proposition}
\begin{proof}
	We have $v \in \Pref(\{x, \te(x)\}^+)$ and $u \in \Suff(\{x, \te(x)\}^+)$ because $vu \in \{x, \te(x)\}^n$. 
	Proposition~\ref{prop:pref_suff_split} implies that either the first property holds or $uv \in \{x^n, \te(x)^n\}$. 
	Here we consider only the case when $uv = x^n$. 
	Then $u = x^i x_p$ and $v = x_s x^{n-i-1}$ for some $1 \le i \le n$ and $x_p, x_s \in \Sigma^+$ with $x = x_p x_s$. 
	Thus, we have $x_p vux_s = x^{n+1}$, from which can deduce $vu = \te(x)^n$ with the aid of Theorem~\ref{thm:overlap3} and the fact that $x$ cannot be a proper infix of its square. 
\end{proof}

\begin{proposition}\label{prop:clean_split}
	Let $x \in \Sigma^+$ be a $\te$-primitive word, and $p, q \in \Sigma^+$ be $\te$-palindromes. 
	If $pq$ is primitive, and $pq = x_1 \cdots x_n$ for some $n \ge 2$ and $x_1, \ldots, x_n \in \{x, \te(x)\}$, 
	then there are integers $k, m \ge 1$ such that $n = 2m$, $p = x_1 \cdots x_{2k}$, and $q = x_{2k+1} \cdots x_{2m}$. 
\end{proposition}
\begin{proof}
	It is clear from $pq = x_1 \cdots x_n$ that $p \in \Pref(\{x, \te(x)\}^+)$ and $q \in \Suff(\{x, \te(x)\}^+)$. 
	Since both $p$ and $q$ are $\te$-palindromes, these mean that $p \in \Suff(\{x, \te(x)\}^+)$ and $q \in \Pref(\{x, \te(x)\}^+)$. 
	Hence we can apply Proposition~\ref{prop:pref_suff_split} to obtain $p = x_1 \cdots x_{i}$ and $q = x_{i+1} \cdots x_{n}$ for some $i$ (since $pq$ is primitive, the case $x_1 = \cdots = x_n$ is impossible).

	The integer $i$ has to be even ($i = 2k$ for some $k \ge 1$). 
	Suppose not, then $p$ being a $\te$-palindrome implies that $x_{(i+1)/2}$ is a $\te$-palindrome, and hence so is $x$. 
	As a result, $pq = x^n$ but this contradicts the assumption that $pq$ is primitive. 
	Similarly, $n-i$ proves to be even, too, and we obtain $n=2m$.
\end{proof}

The next two corollaries follow from Proposition~\ref{prop:clean_split}. 
The first one provides us with a sufficient condition for a primitive word that is a catenation of two non-empty $\te$-palindromes to be $\te$-primitive. 

\begin{corollary}\label{cor:prime-te-prime}
	For non-empty $\te$-palindromes $p, q$, if $pq$ is primitive but there does not exist any $x$ such that $p, q \in \{x, \te(x)\}^+$, then $pq$ is $\te$-primitive. 
\end{corollary}

\begin{corollary}\label{cor:pq2_te-prime}
	Let $p, q$ be non-empty $\te$-palindromes such that $pq$ is primitive. 
	Then some word in $\{p, q\}^{+}$ is $\te$-primitive if and only if $pq$ is $\te$-primitive. 
\end{corollary}
\begin{proof}
	The converse implication is trivial because $pq \in \{p, q\}^+$. 
	The direct implication can be proved by considering its contrapositive, which is immediately given by Proposition~\ref{prop:clean_split}. 
\end{proof}

Note that in the statement of Corollary \ref{cor:pq2_te-prime} we cannot replace the quantifier ``some'' with ``all''. 
A trivial example is $(pq)^2 \in \{p, q\}^+$, which is not even primitive. 
We can also provide a non-trivial example as follows: 

\begin{example}
	Let $\te$ be the mirror image over $\{a, b\}^*$, $p = a$, and $q = baaab$. 
	It is clear that $pq = abaaab$ is $\te$-primitive. 
	On the other hand, $qppp = (baaa)^2 \in \{p, q\}^+$ is not even primitive. 
\end{example}

Corollary~\ref{cor:pq2_te-prime} gives a further corollary about the case in which a word obtained from a $\te$-primitive word by cyclic permutation remains $\te$-primitive.

\begin{corollary}\label{cor:conj_pal}
	For two non-empty $\te$-palindromes $p, q$, if $pq$ is $\te$-primitive, then $qp$ is $\te$-primitive. 
\end{corollary}
\begin{proof}
	Since $pq$ is $\te$-primitive, it is primitive and hence its conjugate $qp$ is also primitive. 
	Applying Corollary~\ref{cor:pq2_te-prime} to $qp$ gives the result.
\end{proof}

Corollary~\ref{cor:conj_pal} gives a partial answer to one of our questions on the preservation of $\te$-primitivity under cyclic permutation. 

Now let us examine the equation $pq = x_1 \cdots x_n$ from a different perspective to get some results useful in Section \ref{sec:ExLS}. 
Here we see that the assumptions considered in Proposition \ref{prop:clean_split}: $pq$ being primitive and both of $p, q$ being a $\te$-palindrome are critical to obtain $p, q \in \{x, \te(x)\}^+$. 

\begin{lemma}\label{lem:pal_x1-xn}
	For a $\te$-primitive word $x \in \Sigma^+$ and $k \ge 2$, let $x_1, x_2, \ldots, x_k \in \{x, \te(x)\}$. 
	If $pz = x_1x_2 \cdots x_k$ for some $\te$-palindrome $p$ and non-empty word $z \in \Sigma^+$ with $|z| < |x|$, then $x_1 = x_2 = \cdots = x_{k-1}$. 
	Moreover, if $z$ is also a $\te$-palindrome, then $x_k = x_{k-1}$. 
\end{lemma}
\begin{proof}
	Due to the length condition on $z$, we can let $x_k = yz$ for some non-empty word $y \in \Sigma^+$. 
	Hence we have $p = x_1x_2 \cdots x_{k-1}y$. 
	Since $p$ is a $\te$-palindrome, $p = \te(y) \te(x_{k-1}) \cdots \te(x_1)$. 
	This means that $\te(x_{k-1}) \cdots \te(x_1)$ is a proper infix of $x_1 \cdots x_k$, and we can say that $x_1 = \cdots = x_{k-1}$ using Theorem~\ref{thm:overlap3} (we can assume $k \ge 3$, since if $k = 2$ the consequence is trivial). 

	Now we consider the additional result when $z = \te(z)$. 
	Without loss of generality, we can assume that $x_1 = x$. 
	So we have $p = x^{k-1}y = \te(y)\te(x)^{k-1}$. 
	Since $|y| < |\te(x)|$, this equation gives $\te(x) = qy$ for some non-empty word $q$. 
	Actually $q$ is a $\te$-palindrome. 
	Indeed, we have $qy \in \Suff(p) = \Suff(x^{k-1}y)$, hence as $|q| < |x|$, $q \in \Suff(x)$. 
	Moreover, by definition, $q \in \Pref(\te(x))$, therefore $\te(q) \in \Suff(x)$ and thus $q$ has to be a $\te$-palindrome.

	Thus, if $x_k = \te(x)$, then $\te(x) = qy = yz$ and hence $\te(x)$ could not be $\te$-primitive due to Proposition~\ref{prop:pali_conjugate}, raising a contradiction. 
\end{proof}

For two $\te$-palindromes $p, q$, a $\te$-primitive word $x$, and $x_1, \ldots, x_k \in \{x, \te(x)\}$ ($k \ge 1$), if $|q| < |x|$, then the equation $pq = x_1 \cdots x_k$ turns into $pq = x^k$ due to Lemma \ref{lem:pal_x1-xn} and its solution is $x = p'q$ for some $\te$-palindrome $p'$ such that $p = x^{k-1}p'$. 
If we replace $q$ in this equation with a word $z$, which is not assumed to be a $\te$-palindrome, and if $k \ge 3$, then we can still find an intriguing non-trivial solution to the equation $pz = x^{k-1}\te(x)$. 

\begin{example}
	Let $p$ be a $\te$-palindrome, $x$ be a $\te$-primitive word, and $z \in \Sigma^+$ with $|z| < |x|$. 
	For some $i \ge 0$, $j \ge 1$, $k \ge 3$, and $\te$-palindromes $r, t$ such that $rt$ is primitive, we can see that $x = [r(tr)^i]^2 (tr)^j$, $p = x^{k-1} r(tr)^i$, and $z = (tr)^j r(tr)^i$ satisfy $pz = x^{k-1}\te(x)$. 
\end{example}

Note that $r$ and $t$ in this example are given by Proposition \ref{prop:th-commute}.  
Further research on the properties of words in $\{r(tr)^i, (tr)^j\}^*$ may shed light on the properties of $\te$-primitive words. 
In Section \ref{subsec:non-trivial_ExLS4}, we will provide some results along this line, such as the ones in Propositions~\ref{prop:rt2_prime_present1} and \ref{prop:rt2_prime_present2}.

	\section{Extended Lyndon-Sch\"{u}tzenberger equation}\label{sec:ExLS}

As an application of the results obtained in Section~\ref{sec:combinatorial}, we address some open cases of the extended Lyndon-Sch\"{u}tzenberger equation in this section. 

For $u, v, w \in \Sigma^+$, the ExLS equation under consideration is of the form 
\[
	u_1 \cdots u_\ell = v_1 \cdots v_n w_1 \cdots w_m,
\] 
where $u_1, \ldots, u_\ell \in \{u, \te(u)\}$, $v_1, \ldots, v_n \in \{v, \te(v)\}$, and $w_1, \ldots, w_m \in \{w, \te(w)\}$, for $\ell, n, m \ge 2$. 
The open cases are $\ell \in \{2, 3, 4\}$ and $m,n\geq 3$ (see Table \ref{tbl:exLS_summary}). 
It suffices to consider the case when both $v$ and $w$ are $\te$-primitive; otherwise we simply replace them with their $\te$-primitive roots and increase the parameters $n$ and $m$. 
The words $v_1 \cdots v_n$ and $w_1 \cdots w_m$ being symmetric with respect to their roles in the equation, it is also legitimate to assume that $|v_1 \cdots v_n| \ge |w_1 \cdots w_m|$. 

Throughout Subsections \ref{subsec:ExLS4_setting} to \ref{subsec:ExLS4_01}, we prove that the triple $(4, \ge 3, \ge 3)$ imposes $\te$-periodicity. 
First of all, in Subsection~\ref{subsec:ExLS4_setting}, the problem which we actually work on is formalized as Problem~\ref{prob:main}, and we solve some special instances of ExLS equation to which the application of the generalized Fine and Wilf's theorem (Theorem~\ref{thm:exFWlcm}) immediately proves the existence of a word $t$ satisfying $u, v, w \in \{t, \te(t)\}^+$. 
We call such instances {\it trivial ExLS equations}. 
In Subsection~\ref{subsec:non-trivial_ExLS4}, we provide additional conditions which can be assumed for non-trivial ExLS equations. 
Several lemmas and propositions are also proved there. 
They are interesting in their own and our proof techniques for them probably include various applications beyond the investigation on the non-trivial ExLS equations in Subsection \ref{subsec:ExLS4_00} (the case when $u_2 = u_1$) and Subsection \ref{subsec:ExLS4_01} (the case when $u_2 \neq u_1$). 
In each of these subsections, we analyze four cases depending on the values of $u_3$ and $u_4$ one at a time. 
All of these proofs merely consist of direct applications of the results obtained so far and in Subsection~\ref{subsec:non-trivial_ExLS4}. 

In Subsection~\ref{subsec:ExLS3}, we prove that for $n, m \ge 2$, the triple $(3, n, m)$ does not impose $\te$-periodicity. 
We provide several (parametrized) examples which verify that for some specific values of $n, m$, the triple $(3, n, m)$ does not impose $\te$-periodicity. 
Our survey will expose complex behaviors of $(3, n, m)$ ExLS equations.

	\subsection{Problem setting for the ExLS equation $\ell = 4$}
	\label{subsec:ExLS4_setting}

Taking the assumptions mentioned above into consideration, the problem which we are addressing is described as follows: 

\begin{problem}\label{prob:main}
	Let $u, v, w \in \Sigma^+$ and integers $n, m \ge 3$. 
	Let $u_1, u_2, u_3, u_4 \in \{u, \te(u)\}$, $v_1, \ldots, v_n \in \{v, \te(v)\}$, and $w_1, \ldots, w_m \in \{w, \te(w)\}$. 
	Does the equation $u_1u_2u_3u_4 = v_1 \cdots v_n w_1 \cdots w_m$ imply $u, v, w \in \{t, \te(t)\}^+$ for some $t \in \Sigma^+$ under all of the following conditions? 
	\begin{arabiclist}
	\item\label{cond:te-primitive}	$v$ and $w$ are $\te$-primitive, 
	\item\label{cond:symmetry}	$|v_1 \cdots v_n| \ge |w_1 \cdots w_m|$, 
	\item\label{cond:fixed_word}	$u_1 = u$, $v_1 = v$, and $w_m = w$, 
	\item\label{cond:length}	$|v|, |w| < |u|$.  
	\end{arabiclist}
	The condition~\ref{cond:symmetry} means that $2|u| \le n|v|$. 
	Besides, the condition~\ref{cond:length} follows from the conditions~\ref{cond:te-primitive} and~\ref{cond:symmetry} as shown in the next lemma. 
\end{problem}

\begin{lemma}\label{lem:uv_length}
	Let $u, v, w \in \Sigma^+$ such that $v$, $w$ are $\te$-primitive. 
	If $u_1u_2u_3u_4 = v_1 \cdots v_n w_1 \cdots w_m$ for some $n, m \ge 3$, $u_1, u_2, u_3, u_4 \in \{u, \te(u)\}$, $v_1, \ldots, v_n \in \{v, \te(v)\}$, and $w_1, \ldots, w_m \in \{w, \te(w)\}$, then $|v| < |u|$ and $|w| < |u|$. 
\end{lemma}
\begin{proof}
	Due to Condition~\ref{cond:symmetry}, $|v_1 \cdots v_n| \ge |w_1 \cdots w_m|$. 
	This means that $m|w| \leq 2|u|$, which in turn implies $|w| \le \frac{2}{3}|u|$ because $m \geq 3$. 
	Thus $|w| < |u|$. 

	Now suppose that the ExLS equation held with $|v| \ge |u|$. 
	Then $v_1 \cdots v_n$ is a prefix of $u_1u_2u_3u_4$ of length at least $3|v| \geq 2|v|+|u|$, and hence  $u, v \in \{t, \te(t)\}^+$ for some $\te$-primitive word $t \in \Sigma^+$ due to Theorem~\ref{thm:exFWgcd}. 
	Unless $|v| = |u|$, we reach the contradiction that $v$ would not be $\te$-primitive. 
	Even if $|v| = |u|$, we have $u_4 = w_1 \cdots w_m$. 
	Therefore $v_1 = u_1$ could not be $\te$-primitive. 
\end{proof}

The next lemma reduces the number of steps required to prove a positive answer to Problem~\ref{prob:main}. 

\begin{lemma}\label{lem:two_enough}
	Under the setting of Problem~\ref{prob:main}, if $u, v \in \{t, \te(t)\}^+$ for some $t \in \Sigma^+$, then $w \in \{t, \te(t)\}^+$. 
\end{lemma}

In fact, we can say more strongly that if two of $u, v, w$ are proved to be in $\{t, \te(t)\}^+$ for some $t$, then the other one is also in this set. 

First of all, we distinguish the case in which the existence of such $t$ that $u, v, w \in \{t, \te(t)\}^+$ is trivial due to the generalized Fine and Wilf theorem (Theorem \ref{thm:exFWlcm}). 

\begin{theorem}\label{thm:trivial}
	Under the setting of Problem~\ref{prob:main}, if there exists an index $i$, $1 \le i \le n$, such that $u_1u_2 = v_1 \cdots v_i$, then $u, v, w \in \{t, \te(t)\}^+$ for some word $t \in \Sigma^+$. 
\end{theorem}
\begin{proof}
	Since $v$ is assumed to be $\te$-primitive, Theorem~\ref{thm:exFWlcm} implies $u \in \{v, \te(v)\}^+$. 
	Then $w \in \{v, \te(v)\}^+$ due to Lemma~\ref{lem:two_enough} (in fact, $w \in \{v, \te(v)\}$ because $w$ is also assumed to be $\te$-primitive). 
\end{proof}

If a given $(4, n, m)$ ExLS equation satisfies the condition in Theorem~\ref{thm:trivial}, then we say that this equation is {\it trivial}. 
Before initiating our study on non-trivial ExLS equations, we provide one important condition which makes the equation trivial according to the generalized Fine and Wilf theorem (Theorem \ref{thm:exFWgcd}). 

\begin{proposition}\label{prop:longenough}
	Under the setting of Problem~\ref{prob:main}, if $n|v| \ge 2|u|+|v|$, then the equation is trivial. 
\end{proposition}
\begin{proof}
	We can employ Theorem~\ref{thm:exFWgcd} to obtain $u, v \in \{t, \te(t)\}^+$ for some $t \in \Sigma^+$. 
	In fact, $t$ is either $v$ or $\te(v)$ because $v$ is assumed to be $\te$-primitive. 
	Hence we can find such $i$ stated in Theorem~\ref{thm:trivial}, and by definition this equation is trivial. 
\end{proof}

	\subsection{Non-trivial $(4, \ge 3, \ge 3)$ ExLS equations and related combinatorial results}
	\label{subsec:non-trivial_ExLS4}

Now we shift our attention to the non-trivial $(4, \ge 3, \ge 3)$ ExLS equation. 
What we will actually prove here is that under the setting of Problem~\ref{prob:main}, any {\it non-trivial} equation cannot hold. 
Along with Theorem~\ref{thm:trivial}, this implies that $(4, \ge 3, \ge 3)$ imposes $\te$-periodicity. 

>From this theorem and Proposition~\ref{prop:longenough}, the equation is {\it non-trivial} if and only if $(n-1)|v| < 2|u| < n|v|$. 
Thus, the next proposition, which was proposed in~\cite{CCKS09} to decrease the amount of case analyses for the $(5, \ge 3, \ge 3)$ ExLS equation, is still available for the investigation of non-trivial $(4, \ge 3, \ge 3)$ ExLS equations. 

\begin{proposition}[\cite{CCKS09}]\label{prop:CCKS09-a}
	Let $u, v \in \Sigma^+$ such that $v$ is $\te$-primitive, $u_2, u_3 \in \{u, \te(u)\}$, and $v_2, \ldots, v_n \in \{v, \te(v)\}$ for some integer $n \ge 3$. 
	If $vv_2 \cdots v_n \in \Pref(uu_2u_3)$ and $(n-1)|v| < 2|u| < n|v|$, then there are only two possible cases. 
	\begin{arabiclist}
	\item	$u_2 = \te(u)$: and $v_2 = \cdots = v_n = v$ with $u\te(u) = (pq)^{n-1}p$ and $v = pq$ for some non-empty $\te$-palindromes $p, q$.  
	\item	$u_2 = u$: $n$ is even, $v_2 = \cdots = v_{n/2} = v$, and $v_{n/2+1} = \cdots = v_n = \te(v)$ with $v = r(tr)^i (rt)^{i+j}r$ and $u = v^{n/2-1} r(tr)^i (rt)^j$ for some $i \ge 0$, $j \ge 1$, and non-empty $\te$-palindromes $r, t$ such that $rt$ is primitive. 
	\end{arabiclist}
\end{proposition}

This proposition helps in proving that non-trivial $(4, \ge 3, \ge 3)$ ExLS equations verify the one more condition that $|v| \neq |w|$ as shown in the next proposition. 

\begin{proposition}\label{prop:different_length}
	Non-trivial ExLS equations under the setting of Problem~\ref{prob:main} imply $|v| \neq |w|$. 
\end{proposition}
\begin{proof}
	Suppose that the equation were non-trivial with $|v| = |w|$. 
	Combining $|v|=|w|$ and the non-trivial length condition together implies $m = n-1$ and furthermore the border between $u_2$ and $u_3$ splits $v_{n}$ into exactly halves. 
	Hence if $u_3 = \te(u_2)$, then $v_n = x\te(x)$ for some $x \in \Sigma^+$, contradicting the $\te$-primitivity of $v$. 
	Besides, due to the condition~\ref{cond:length} of Problem~\ref{prob:main}, if $u_4 = \te(u_1)$, then $w = \te(v)$, and hence $u_1u_2u_3u_4 \in \{v, \te(v)\}^+$. 
	Taking $(n-1)|v| < 2|u| < n|v|$ into account, this implies that $v$ is not $\te$-primitive, raising a contradiction. 
	Therefore, the only possible solutions verify $u_3 = u_2$ and $u_4 = u_1 = u$. 

	If $u_2 = u_3 = u$, then according to Proposition~\ref{prop:CCKS09-a}, $n$ is even, and by substituting the representations of $u$ and $v$ given there into $u^4 = v^{n/2}\te(v)^{n/2} w_1 \cdots w_m$, we obtain that $w_1 \cdots w_m = (tr)^j [r(tr)^i r(tr)^{i+j}]^{n/2-1} [r(tr)^{i+j}r(tr)^i]^{n/2-1} (rt)^j$, which is a $\te$-palindrome of even length. 
	Since $w$ is $\te$-primitive, $m$ has to be even (Lemma~\ref{lem:pali_even}).
	It is however impossible because $m = n-1$ and $n$ is even. 
 
	If $u_2 = u_3 = \te(u)$, then Proposition~\ref{prop:CCKS09-a} gives $v = pq$ and $u_1u_2 = u\te(u) = (pq)^{n-1}p$ for some $\te$-palindromes $p, q \in \Sigma^+$. 
	Note that the left side of the ExLS equation is as long as its right side ($4|u| = n|v| + m|w| = (2n-1)|pq|$). 
	Substituting $2|u| = (n-1)|pq| + |p|$ into this yields $|p| = |q|$ and it in turn implies that both $p$ and $q$ are of even length. 
	Let $p = p'\te(p')$ and $q = q'\te(q')$ for some $p', q' \in \Sigma^+$ of the same length. 
	Then $u_1 = u$ ends with either $\te(p')qp'$ or $\te(q')pq'$, and so $w_m$ is either of them. 
	However, neither is $\te$-primitive. 
	This contradiction proves that the equation is trivial. 
\end{proof}

Supposing that some non-trivial $(4, \ge 3, \ge 3)$ ExLS equation held, the next claim would follow from this proposition. 
Although our conclusion in this section will prove that this claim cannot hold, the equation proposed there, $u_3u_4 = q w_1 \cdots w_m$, or more generally the relation $q w_1 \cdots w_m \in \{u, \te(u)\}^{\ge 2}$ provides in its own right challenging themes. 

\begin{claim}\label{claim:non-trivial}
	Under the setting of Problem \ref{prob:main}, if the ExLS equation were non-trivial, then we would have $u_3u_4 = q w_1 \cdots w_m$ for some non-empty $\te$-palindrome $q$. 
\end{claim}
\begin{proof}
	According to the presentations of $u$ and $v$ given in Proposition~\ref{prop:CCKS09-a}, 
	if $u_2 = \te(u)$, then $u\te(u)q = v^n$ and hence $u_3u_4 = q w_1 \cdots w_m$; 
	otherwise, $uu [r(tr)^i]^2 = v^{n/2} \te(v)^{n/2}$ so that $u_3u_4 = [r(tr)^i]^2 w_1 \cdots w_m$. 
	Since $q, r, t$ are $\te$-palindromes, this claim holds. 
\end{proof}

As we shall see soon in Claim~\ref{claim:u3nequ4}, the next lemma is of use when considering non-trivial ExLS equations with $u_3 \neq u_4$, that is, $u_3u_4$ being a $\te$-palindrome. 

\begin{lemma}\label{lem:pali_pref_pali}
	Let $p, q$ be non-empty $\te$-palindromes and let $w$ be a $\te$-primitive word. 
	For some $k \ge 1$ and words $w_1, \ldots, w_k \in \{w, \te(w)\}$, if $p = q w_1 \cdots w_k$ holds, then either $p, q \in \{w, \te(w)\}^+$ or $w_1 = \cdots = w_k$.
\end{lemma}
\begin{proof}
	First we prove that $q \in \Suff((w_1 \cdots w_k)^+)$. 
	Since $w_1 \cdots w_k \in \Suff(p)$, $p$ being a $\te$-palindrome implies $\te(w_1 \cdots w_k) \in \Pref(p)$. 
	Thus if $|q| \le k|w|$, then $q \in \Pref(\te(w_1 \cdots w_k))$, that is, $q \in \Suff(w_1 \cdots w_k)$ and we are done. 
	Otherwise, $w_1 \cdots w_k \in \Suff(q)$ so that $(w_1 \cdots w_k)^2 \in \Suff(p)$. 
	By repeating this process, eventually we will find some integer $i \ge 1$ such that $q \in \Suff((w_1 \cdots w_k)^i)$. 

	If $q \in \{w, \te(w)\}^+$, then obviously $p \in \{w, \te(w)\}^+$. 
	Otherwise, let $q = w'w_{j+1}\cdots w_k (w_1 \cdots w_k)^i$ for some $1 \leq j \leq k$ and $i \geq 0$, where $w'$ is a non-empty proper suffix of $w_{j}$. 
	Then, $p = w'w_{j+1}\cdots w_k (w_1 \cdots w_k)^{i+1}$ overlaps in a non-trivial way with $p = \te(p) = (\te(w_k)\cdots \te(w_1))^{i+1}\te(w_k)\cdots \te(w_{j+1})\te(w')$, and Theorem~\ref{thm:overlap3} implies that $w_1 = \cdots = w_k$. 
\end{proof}

\begin{claim}\label{claim:u3nequ4}
	Under the setting of Problem \ref{prob:main}, if the ExLS equation were non-trivial and $u_3 \neq u_4$, then $w_1 = \cdots = w_m = w$ and $u_3u_4 \in \Suff(w^+)$. 
\end{claim}
\begin{proof}
	We have $u_3u_4 = x w_1 \cdots w_m$ for some non-empty $\te$-palindrome $x \in \Sigma^+$ due to Proposition~\ref{prop:CCKS09-a}.
	As suggested before, we can employ Lemma~\ref{lem:pali_pref_pali} to get either $x, u_3u_4 \in \{w, \te(w)\}^+$ or $w_1 = \cdots = w_m$.  
	In the first case, Theorem~\ref{thm:exFWlcm} implies $u \in \{w, \te(w)\}^+$ because $w$ is assumed to be $\te$-primitive. 
	Then the ExLS equation in turn implies that $v_1 \cdots v_n \in \{w, \te(w)\}^+$ and hence $v \in \{w, \te(w)\}$ for the same reason. 
	As a result the equation would be trivial. 
	Consequently $w_1 = \cdots = w_m$.
\end{proof}

The main strategy used in the analyses of non-trivial ExLS equations is to split $w_1 \cdots w_m$ into smaller components which are still in $\{w, \te(w)\}^+$, until we reach a contradiction. 
The split is mainly achieved by Propositions~\ref{prop:pali_split} and \ref{prop:clean_split}. 
Note that the word to which Proposition~\ref{prop:clean_split} is applied must be primitive.  
The next two lemmas work for this purpose in Subsection~\ref{subsec:ExLS4_00}, but we provide them in more general form. 
An interesting point is that Lyndon and Sch\"{u}tzenberger's original result (Theorem~\ref{thm:original}) plays an essential role in their proofs; hence for the ExLS equation.

\begin{proposition}\label{prop:rt2_prime_present1}
	Let $r, t \in \Sigma^+$ such that $rt$ is primitive. 
	For any $i \ge 0$, $j, k \ge 1$, and $n \ge 2$, $(tr)^j [(r(tr)^i)^n (tr)^j]^k$ is primitive.
\end{proposition}
\begin{proof}
	Suppose that the given word were not primitive; namely, for some $\ell \ge 2$ and a primitive word $x$, let $(tr)^j [(r(tr)^i)^n (tr)^j]^k = x^\ell$. 
	Catenating $(r(tr)^i)^n$ to the left to the both sides of this equation gives $[(r(tr)^i)^n (tr)^j]^{k+1} = (r(tr)^i)^n x^\ell$.
	As $k \geq 1$ and $n, \ell \geq 2$, we can apply Theorem \ref{thm:original} to this equation to obtain $\rho((r(tr)^i)^n (tr)^j) = \rho(r(tr)^i) = x$.
	Using Lemma~\ref{lem:rt_rootshare}, one can obtain $\rho((tr)^j) = x$, and furthermore, $\rho(tr) = x$. 
	Combining this with $\rho(r(tr)^i) = x$ gives us $\rho(r) = \rho(t)$ and hence $rt$ would not be primitive, which contradicts the hypotheses.
\end{proof}

\begin{proposition}\label{prop:rt2_prime_present2}
	Let $r, t \in \Sigma^+$ such that $rt$ is primitive. 
	For any $i \ge 0$, $j, k, m \ge 1$, $(tr)^j [(r(tr)^i)^m (tr)^j]^{k-1} (r(tr)^i)^{m-1} (rt)^j$ is primitive.
\end{proposition}
\begin{proof}
	Suppose that we had $(tr)^j [(r(tr)^i)^m (tr)^j]^{k-1} (r(tr)^i)^{m-1} (rt)^j = x^\ell$ for some primitive word $x$ and $\ell \ge 2$. 
	Catenating $(r(tr)^i)^{m+1}$ to the right to the both sides of this equation gives $[(tr)^j (r(tr)^i)^m]^{k+1} = x^\ell (r(tr)^i)^{m+1}$.
	Now as in the proof of Proposition~\ref{prop:rt2_prime_present1}, we reach the contradicting conclusion that $rt$ is not primitive.
\end{proof}

There are some results which can be used for the splitting strategy, once we apply Proposition~\ref{prop:CCKS09-a} to non-trivial ExLS equations with $u_1 \neq u_2$, which will be considered in Subsection \ref{subsec:ExLS4_01}. 
As before, they are provided in more general form than required for the purpose.

\begin{lemma}\label{lem:zp-nonprime}
	Let $z, w \in \Sigma^+$ with $|z| < |w|$ and let $p$ be a $\te$-palindrome. 
	If $zp = w^n$ for some $n \ge 2$, then $z = \te(z)$. 
\end{lemma}
\begin{proof}
	Let $w = zy$ for some $y \in \Sigma^+$. 
	Then $p = y(zy)^{n-1}$, from which we can obtain $y = \te(y)$ and $z = \te(z)$ because $p = \te(p)$ and $n-1 \ge 1$. 
\end{proof}

\begin{proposition}\label{prop:uteuqnu}
	Let $x$ be a $\te$-primitive word, $u \in \Sigma^+$, and $q$ be a non-empty $\te$-palindrome. 
	If for some $n \ge 2$ and $\ell \ge 1$, $u[\te(u)q^n u]^\ell \in \{x, \te(x)\}^{\ge 2}$, then $u, q \in \{x, \te(x)\}^+$. 
\end{proposition}
\begin{proof}
	Let $u[\te(u)q^n u]^\ell = x_1 \cdots x_m$ for some $m \ge 2$ and $x_1, \ldots, x_m \in \{x, \te(x)\}$. 
	Let $u = x_1 \cdots x_{k-1}z_1$ and $[\te(u)q^n u]^\ell = z_2 x_{k+1} \cdots x_m$ for some $1 \le k \le m$ with $x_k = z_1 z_2$ and $z_1 \neq \lambda$, i.e. $|z_2| < |x|$. 
	If $z_2 = \lambda$, then $u, [\te(u)q^n u]^\ell \in \{x, \te(x)\}^+$. 
	Lemma~\ref{lem:rt_th-rootshare} implies $\te(u)q^n u \in \{x, \te(x)\}^+$ and the same lemma further gives $q^n \in \{x, \te(x)\}^+$, that is, $q \in \{x, \te(x)\}^+$. 

	Now we prove that $z_2$ cannot be non-empty. 
	Without loss of generality, we assume $x_m = x$. 
	So suppose $z_2 \neq \lambda$ ($0 < |z_1| < |x|$).
	We can apply Lemma~\ref{lem:pal_x1-xn} to $z_1[\te(u)q^nu]^\ell = x_k \cdots x_m$ to get $x_{k+1} = \cdots = x_m = x$ because $[\te(u)q^n u]^\ell$ is a $\te$-palindrome and $|z_1| < |x|$. 
	Thus if $|x| \le |u|$, then $|z_2| < |u|$ and so $[\te(u)q^n u]^\ell = z_2 x^{k-1}$ gives $x \in \Suff(u)$ and hence $\te(x) \in \Pref(\te(u))$. 
	These further imply that $x \in \Suff(x_1 \cdots x_{k-1}z_1)$ and $\te(x) \in \Pref(z_2 x_{k+1} \cdots x_m)$. 
	Thus $x\te(x)$ is a proper infix of $x_{k+1}x_kx_{k-1}$, which is in contradiction with the $\te$-primitivity of $x$ by Theorem~\ref{thm:overlap3}. 

	Therefore, $|x| > |u|$, which means $k = 1$, that is, we have $x_2 = \cdots = x_m = x$. 
	Note that $x \neq \te(x)$ must hold because of $z_2 x^{m-1}$ being a $\te$-palindrome, $0 < |z_2| < |x|$ and $x$ is primitive (and cannot be a proper infix of its square). 
	If $x_1 = \te(x)$, then $u \in \Pref(\te(x)) \cap \Suff(x)$ holds and so $u = \te(u)$. 
	Now Lemma~\ref{lem:pal_x1-xn} would imply $x_1 = x$, which contradicts $x \neq \te(x)$. 
	Otherwise ($x_1 = x$), $u[\te(u)q^n u]^\ell = x^m$ and from this Lemma~\ref{lem:zp-nonprime} derives $u = \te(u)$. 
	Then we have $u(uq^nu)^\ell = x^m$; in other words, $(uq^nu)^{\ell+1}$ and $x^m$ share a suffix of length at least $\eta = \max(m|x|,\ell|uq^nu|)$.  
	If $\ell \ge 2$, then $\eta \geq |x| + |uq^nu|$, and the Fine and Wilf theorem implies $\rho(uq^nu) = x$. 
	With $u(uq^nu)^\ell = x^m$, this implies $\rho(u) = x$.
	However, this contradicts $|u| < |x|$. 
	If $\ell = 1$, then $uuq^nu = x^m$. 
	Using cyclic permutation, we obtain $u^3q^n = x'^m$, where $x'$ is a conjugate of $x$. 
	This is of the form of LS equation, and Theorem~\ref{thm:original} concludes $\rho(u) = \rho(q) = x'$. 
	Now we reached the same contradiction because $|x'| = |x|$. 
\end{proof}

\begin{lemma}\label{lem:u2qwm2}
	Let $w$ be a $\te$-primitive word, and $w_1, \ldots, w_m \in \{w, \te(w)\}$ for some $m \ge 2$. 
	Let $u, q \in \Sigma^+$ such that $q$ is a $\te$-palindrome with $|q| < |u|$. 
	If $u^2 = q w_1 \cdots w_m$, then either $u, q \in \{w, \te(w)\}^+$ or $u = qr$ for some non-empty $\te$-palindrome $r$. 
\end{lemma}
\begin{proof}
	It is trivial that the case $u, q \in \{w, \te(w)\}^+$ is possible. 
	Hence assume that $u, q \not\in \{w, \te(w)\}^+$.
	Without loss of generality, we can also assume that $w_m = w$.
	Let $u = qr$ for some $r \in \Sigma^+$. 
	Then $rqr = w_1 \cdots w_m$. 
	We prove that $r$ is a $\te$-palindrome.
	Let $r = w_1 \cdots w_{k-1} z_1 = z_2 w_{m-k+2} \cdots w_m$ for some $k \ge 1$, where $z_1 \in \Pref(w_{k})$ and $z_2 \in \Suff(w_{m-k+1})$ with $|z_1| = |z_2| < |w|$. 
	If $z_1 = \lambda$, then $r \in \{w, \te(w)\}^+$ and then $rqr = w_m \cdots w_1$ implies $q \in \{w, \te(w)\}^+$ by Lemma~\ref{lem:rt_th-rootshare}, but this contradicts the assumption. 
	Thus $z_1 \neq \lambda$.
	Then we have two cases, $k \ge 2$ and $k = 1$. 
	Lemma~\ref{lem:overlap2} (for $k = 2$) or Theorem~\ref{thm:overlap3} (for $k \ge 3$) works to give $w_1 = \cdots = w_{k-1} = \te(w)$ and $w_{m-k+2} = \cdots = w_m = w$. 
	Thus, $z_2 = \te(z_1)$ and hence $r = \te(r)$. 
	Even for $k = 1$, if $w_1 \neq w_m$, then $r \in \Pref(\te(w)) \cap \Suff(w)$ so that $r = \te(r)$. 
	Otherwise $w = r q_p = q_s r$ for some $q_p \in \Pref(q)$ and $q_s \in \Suff(q)$. 
	Since $q = \te(q)$, $q_s = \te(q_p)$ so that we have $r q_p = \te(q_p) r$. 
	According to Proposition~\ref{prop:th-commute}, $r = \te(r)$. 
\end{proof}

\begin{proposition}\label{prop:u2qwm3odd}
	Let $w$ be a $\te$-primitive word, and $w_1, \ldots, w_m \in \{w, \te(w)\}$ for some odd integer $m \ge 3$. 
	Let $u, q \in \Sigma^+$ such that $q$ is a $\te$-palindrome with $|q| < |u|$. 
	If $u^2 = q w_1 \cdots w_m$, then $w = \te(w)$. 
	If additionally $|u| \ge 2|q|$ holds, then $\rho(u) = \rho(q) = w$. 
\end{proposition}
\begin{proof}
	Lemma~\ref{lem:u2qwm2} implies that either $q, u \in \{w, \te(w)\}^+$ or $u = qr$ for some non-empty $\te$-palindrome $r$. 
	In the former case, let $u \in \{w, \te(w)\}^k$ for some $k \ge 1$ and we can see $q \in \{w, \te(w)\}^{2k-m}$ and $2k-m$ is odd because $m$ is odd. 
	Then $q = \te(q)$ implies $w = \te(w)$, and hence $u, q \in w^+$. 
	In the latter case, we have $rqr = w_1 \cdots w_m$. 
	This implies $w_{(m+1)/2} = \te(w_{(m+1)/2})$ (i.e, $w = \te(w)$) because $rqr$ is a $\te$-palindrome and $m$ is odd. 

	Now we consider the additional hypothesis $|u| \ge 2|q|$. 
	Since $2|u| = |q| + m|w|$, $|u| = (|q|+m|w|)/2 \geq 2|q|$, which leads to $|q| \leq \frac{1}{3}m|w|$.
	As seen above, $rqr = w^m$, hence $|r| = (m|w|-|q|)/2 \geq \frac{1}{3}m|w| \geq |w|$ as $m \ge 3$. 
	With this, the equation $rqr = w^m$ gives $r = w^k w_p' = w_s' w^k$ for some $k \ge 1$, $w_p' \in \Pref(w)$, and $w_s' \in \Suff(w)$. 
	Since $w$ is primitive, $w_p'$ and $w_s'$ have to be empty. 
	Consequently $\rho(r) = \rho(q) = w$ and hence $\rho(u) = w$ by using Lemma~\ref{lem:rt_th-rootshare}. 
\end{proof}


	\subsection{ExLS equation of the form $u^2u_3u_4 = v_1 \cdots v_n w_1 \cdots w_m$}
	\label{subsec:ExLS4_00}

In this subsection, we prove that an ExLS equation of the form $u^2u_3u_4 = v_1 \cdots v_n w_1 \cdots w_m$ implies that $u, v, w \in \{t, \te(t)\}^+$ for some $t \in \Sigma^+$. 
We have already seen that for this purpose it suffices to show that any non-trivial equation of this form cannot hold. 
Recall that we assumed $u_1 = u$, $v_1 = v$, and $w_m = w$, and that Proposition~\ref{prop:different_length} allows us to assume $|v| \neq |w|$.

We can apply Proposition~\ref{prop:CCKS09-a} to the non-trivial equation to obtain that $n$ is an even integer except 2, $v_1 = \cdots = v_{n/2} = v$ and $v_{n/2+1} = \cdots = v_n = \te(v)$ (i.e., $v_1 \cdots v_n$ is a $\te$-palindrome), $u = [r(tr)^i r(tr)^{i+j}]^{n/2-1} r(tr)^i (rt)^j$, and $v = r(tr)^i r(tr)^{i+j}$ for some $i \ge 0$, $j \ge 1$, and non-empty $\te$-palindromes $r, t$ such that $rt$ is primitive. 
Actually $rt$ has to be $\te$-primitive due to Corollary~\ref{cor:pq2_te-prime} because $v \in \{r, t\}^+$ is assumed to be $\te$-primitive. 
Let us now study all possible values of $u_3u_4$.

\begin{proposition}
	Under the setting of Problem \ref{prob:main}, if $u_1u_2u_3u_4 = u^4$, then $u, v, w \in \{t, \te(t)\}^+$ for some $t \in \Sigma^+$. 
\end{proposition}
\begin{proof}
	According to the representations of $u$ and $v$ in terms of $r$ and $t$, we obtain
	\[
		w_1 \cdots w_m = (tr)^j [(r(tr)^i)^2 (tr)^j]^{n/2-1} [(rt)^j (r(tr)^i)^2]^{n/2-1} (rt)^j\enspace.
	\] 
	This expression is a $\te$-palindrome of even length and hence $m$ has to be even (Lemma~\ref{lem:pali_even}). 
	Therefore, $w_1 \cdots w_{m/2} = [(tr)^j (r(tr)^i)^2]^{n/2-1} (tr)^j$, and this was proved to be primitive in Proposition~\ref{prop:rt2_prime_present1}.
	Moreover, its right hand side is the catenation of two $\te$-palindromes $p_1 = (tr)^j [r(tr)^i r(tr)^{i+j}]^{n/2-2} r(tr)^i (rt)^j$ and $p_2 = r(tr)^i$. 
	Proposition~\ref{prop:clean_split} gives $p_2 = r(tr)^i \in \{w, \te(w)\}^+$. 
	Furthermore, applying Proposition~\ref{prop:pali_split} to $p_1 p_2 = (tr)^j [r(tr)^i r(tr)^{i+j}]^{n/2-2} r(tr)^i \cdot p_2 \cdot (tr)^j$ gives $(tr)^j \in \{w, \te(w)\}^+$. 
	Finally Lemma~\ref{lem:rt_th-rootshare} derives $r, t \in \{w, \te(w)\}^+$ from $r(tr)^i, (tr)^j \in \{w, \te(w)\}^+$, but this contradicts the $\te$-primitivity of $rt$. 
	As a result, there are no solutions to the non-trivial equation.
\end{proof}

\begin{proposition}\label{prop:LS4_0001}
	Under the setting of Problem \ref{prob:main}, if $u_1u_2u_3u_4 = u^3\te(u)$, then $u, v, w \in \{t, \te(t)\}^+$ for some $t \in \Sigma^+$. 
\end{proposition}
\begin{proof}
	Since $u_4$ is $\te(u)$ instead of $u$, we have $w_1 \cdots w_m = x^2 (r(tr)^i)^2$, where $x = (tr)^j [(r(tr)^i)^2 (rt)^j]^{n/2-1} r(tr)^i (rt)^j$. 
	Claim~\ref{claim:u3nequ4} gives that $w_1 = \cdots = w_m = w$, and hence $w^m = x^2 (r(tr)^i)^2$.
	This is a classical LS equation; thus Theorem~\ref{thm:original} is applicable to conclude that $\rho(x) = \rho(r(tr)^i)$. 
	However, this contradicts the primitivity of $x$ obtained in Proposition~\ref{prop:rt2_prime_present2} because $|x| > |r(tr)^i|$.
\end{proof}

\begin{proposition}
	Under the setting of Problem \ref{prob:main}, if $u_1u_2u_3u_4 = u^2\te(u)u$, then $u, v, w \in \{t, \te(t)\}^+$ for some $t \in \Sigma^+$. 
\end{proposition}
\begin{proof}
	Since $u_3 \neq u_4$, $w_1 = \cdots = w_m = w$ due to Claim~\ref{claim:u3nequ4}. 
	Using the representations of $u$ and $v$ by $r$ and $t$, we can see that $u_3u_4 = \te(u)u$ is equal to both sides of the following equation: 
	\[
		(tr)^j r(tr)^i [(rt)^j (r(tr)^i)^2]^{n/2-1} [(r(tr)^i)^2 (tr)^j]^{n/2-1} r(tr)^i (rt)^j = (r(tr)^i)^2 w^m\enspace.  
	\]
	By catenating $(r(tr)^i)^4$ to the left of both sides, we get $(r(tr)^i)^6 w^m = x^2$, where $x = (r(tr)^i)^2 [(r(tr)^i)^2 (tr)^j]^{n/2-1} r(tr)^i (rt)^j$. 
	Then, Theorem~\ref{thm:original} implies that $\rho(x) = \rho(r(tr)^i) = w$. 
	Since $x$ contains $r(tr)^i$ as its infix, the share of primitive root between $x$ and $r(tr)^i$ gives $\rho(r(tr)^i) = \rho((rt)^j)$. 
	We deduce from this using Lemma~\ref{lem:rt_rootshare} that $rt$ would not be primitive, which contradicts our hypothesis.
\end{proof}

\begin{proposition}
	Under the setting of Problem \ref{prob:main}, if $u_1u_2u_3u_4 = u^2\te(u)^2$, then $u, v, w \in \{t, \te(t)\}^+$ for some $t \in \Sigma^+$. 
\end{proposition}
\begin{proof}
	Recall that $v_1 \cdots v_n$ is a $\te$-palindrome. 
	Since $u^2\te(u)^2$ is a $\te$-palindrome, $\te(w_1 \cdots w_m)$ is one of its prefixes and the assumption $|w_1 \cdots w_m| < |v_1 \cdots v_n|$ implies that $\te(w_1 \cdots w_m) \in \Pref(v_1 \cdots v_n)$.
	Hence $w_1 \cdots w_m \in \Suff(v_1 \cdots v_n)$ and now we have $(w_1 \cdots w_m)^2 \in \Suff(u^2\te(u)^2)$. 

	We prove that this suffix is long enough to apply the extended Fine and Wilf theorem. 
	Since $(n-1)|v| < 2|u|$ and $n \geq 4$, we have $|v| < \frac{2}{3} |u|$ and, in turn, $n|v| < 2|u| + \frac{2}{3} |u| = \frac{8}{3} |u|$.
	From this we obtain $m|w| > \frac{4}{3}|u|$. 
	Then, $2m|w| - (|w|+2|u|) > (2m-1)|w| - \frac{3}{2}m|w| = (\frac{1}{2}m-1)|w| > 0$ since $m \ge 3$. 
	Thus, $u^2 \te(u)^2$ and $(w_m \cdots w_1)^2$ share a suffix of length at least $2|u|+|w|$ and Theorem~\ref{thm:exFWgcd} concludes that $u \in \{w, \te(w)\}^+$ because $w$ is $\te$-primitive. 
	Now it is clear that also $v \in \{w, \te(w)\}^+$, but in fact $v \in \{w, \te(w)\}$ must hold because $v$ is also $\te$-primitive. 
	However this contradicts the assumption that $|v| \neq |w|$.
\end{proof}

	\subsection{ExLS equation of the form $u\te(u)u_3u_4 = v_1 \cdots v_n w_1 \cdots w_m$}
	\label{subsec:ExLS4_01}

Note that in the following propositions, we consider only the non-trivial equations; hence Proposition~\ref{prop:different_length} allows to assume $|v| \neq |w|$.

Using Proposition~\ref{prop:CCKS09-a}, $u\te(u) = (pq)^{n-1}p$ and $v_1 = \cdots = v_n = v = pq$ for some non-empty $\te$-palindromes $p, q$. 
Unlike the case considered before, in the current case $n$ can be odd. 
In fact, if $n$ is odd, then $u = (pq)^{(n-1)/2}y$, where $p = y\te(y)$ for some $y \in \Sigma^+$; while if $n$ is even, then $u = (pq)^{n/2-1}px$, where $q = x\te(x)$ for some $x \in \Sigma^+$. 
Again, we consider the four cases associated with the four possible values of $u_3u_4$. 
The last two, $u_3 = u_4 = u$ and $u_3 = u_4 = \te(u)$, are merged and studied in two separate propositions depending on the parity of $m$ instead. 

\begin{proposition}
	Under the setting of Problem \ref{prob:main}, if $u_1u_2u_3u_4 = u\te(u)u\te(u)$, then $u, v, w \in \{t, \te(t)\}^+$ for some $t \in \Sigma^+$. 
\end{proposition}
\begin{proof}
	In this setting, $u_3u_4 = u\te(u) = q w_1 \cdots w_m$. 
	Since both $u\te(u)$ and $q$ are $\te$-palindromes, we can employ Claim~\ref{claim:u3nequ4} to obtain $w_1 = \cdots = w_m = w$. 
	Now the equation turns into the LS equation $(u\te(u))^2 = v^n w^m$, and hence $\rho(v) = \rho(w)$ due to Theorem~\ref{thm:original}. 
	Both $v$ and $w$ being primitive, this contradicts the assumption $|v| \neq |w|$ and consequently the existence of non-trivial solutions. 
\end{proof}

\begin{proposition}
	Under the setting of Problem \ref{prob:main}, if $u_1u_2u_3u_4 = u\te(u)\te(u)u$, then $u, v, w \in \{t, \te(t)\}^+$ for some $t \in \Sigma^+$. 
\end{proposition}
\begin{proof}
	Recall that $u\te(u) = (pq)^{n-1}p$. 
	Claim~\ref{claim:u3nequ4} implies that $\te(u)u = q w^m$ with $q = w' w^{k-1}$ for some $1 \le k \le m$ and a non-empty proper suffix $w'$ of $w$. 

	{\bf Case 1 ($n$ is odd)}: 
	Then we have $\te(u)u = qw^m = x_s x$, where $x_s = \te(y) q(pq)^{(n-1)/2-1} y$ and $x = \te(y) (pq)^{(n-1)/2} y$; note that $x_s \in \Suff(x)$. 
	One can easily calculate that $|w| = \frac{1}{m}[n|p| + (n-2)|q|]$ and $|x_s| = \frac{1}{2}(n-1)(|p| + |q|)$, and hence $|x_s|-|w| = \frac{(m-2)(n-1)-2}{2m}|p|+\frac{(m-2)(n-1)+2}{2m}|q|$, which is positive because $n, m \ge 3$. 
	Thus we can say that $x^2$ and $w^{m+k}$ share a prefix of length at least $|x|+|w|$ so that by the Fine and Wilf theorem, $\rho(x) = \rho(w) = w$.
	Starting from $\te(y)yqw^m = \te(y)yx_sx = x^2$, we can verify that $2|x|- m|w| = |pq|$, that is, $|pq|$ is a multiple of $|w|$. 
	The suffix of $x$ of length $|pq|$ is $\te(y)qy$, which is $w^j$ for some $j \ge 2$ because $|pq| = |v| \neq |w|$. 
	Therefore, this conjugate of $v$ is not primitive, either. 
	This is a contradiction with the $\te$-primitivity of $v$. 

	{\bf Case 2 ($n$ is even)}: 
	In this case, $u = (pq)^{n/2-1}px$ for some $x \in \Sigma^+$ such that $q = x\te(x)$. 
	Substituting this into $\te(u)u = qw^m$ gives 
	\begin{equation}\label{eq:0110-1}
		[\te(x)px]^{n/2-1} \te(x)p^2 x [\te(x)px]^{n/2-1} = x\te(x) w^m. 
	\end{equation}
	From this equation, we can obtain $x = \te(x)$ and hence $px = xz$ for some $z \in \Sigma^+$. 
	If $|x| \ge |p|$, then Lemma \ref{lem:pq-qp} implies $\rho(x) = \rho(p)$ so that $v = pq = px^2$ would not be primitive. 
	Hence $|x| < |p|$ must hold and under this condition, the solution of $px = xz$ is given by $p = xy$ and $z = yx$ for some $y \in \Sigma^+$. 
	Since $p = \te(p)$, we have $p = xy = \te(y)x$. 
	Proposition \ref{prop:th-commute} gives $x = r(tr)^i$ and $y = (tr)^j$ for some $i \ge 0$, $j \ge 1$, and $\te$-palindromes $r, t$ such that $rt$ is primitive. 
	Both of $r$ and $t$ should be non-empty; otherwise, $\rho(p) = \rho(x)$ and $v = pq = px^2$ would not be primitive.  
	Substituting these into Eq.~(\ref{eq:0110-1}) yields the following equation. 
	\[
		\bigl[(tr)^j r(tr)^i [r(tr)^i r(tr)^{i+j} r(tr)^i]^{n/2-1} \bigr]^2 = w^m. 
	\]
	Since $w$ is $\te$-primitive, this equation means that $m$ has to be even. 
	Then $w^{m/2} = (tr)^j r(tr)^i [r(tr)^i r(tr)^{i+j} r(tr)^i]^{n/2-1}$. 
	By catenating $(r(tr)^i)^2$ from the left to the both sides of this equation, we obtain an LS equation $[r(tr)^i]^2 w^{m/2} = [r(tr)^i r(tr)^{i+j} r(tr)^i]^{n/2}$. 
	Theorem \ref{thm:original} gives $\rho(r(tr)^i) = \rho(r(tr)^i r(tr)^{i+j} r(tr)^i)$ and Lemma \ref{lem:rt_rootshare} reduces it to $\rho(r) = \rho(t)$, but this contradicts the primitivity of $pq = r(tr)^{i+j} (r(tr)^i)^2$. 
\end{proof}

\begin{proposition}
	Under the setting of Problem \ref{prob:main}, if $u_1u_2 = u\te(u)$, $u_3 = u_4$, and $m$ is odd, then $u, v, w \in \{t, \te(t)\}^+$ for some $t \in \Sigma^+$. 
\end{proposition}
\begin{proof}
	We have $u_3 u_4 = q w_1 \cdots w_m$. 
	Since $u_3 = u_4$ and $|q| < |u|$, we can employ Proposition~\ref{prop:u2qwm3odd} to obtain $w = \te(w)$. 
	Moreover, when $n \geq 5$, we have $|u| \ge 2|q|$ and the proposition also gives $\rho(u_3) = \rho(q) = w$. 
	Since $w = \te(w)$, we can see that $\rho(u) = w$. 
	Then $\rho(p) = w$ because $\rho(u) = \rho(q) = w$ and $pq \in \Pref(u)$. 
	However, $\rho(p) = \rho(q)$ means that $v = pq$ would not be even primitive.
	Therefore in the following let $n$ be either 3 or 4. 

	First we consider the case when $u_3 = u$. 
	Then we have either $(pqy)^2 = qw^m$ (when $n = 3$) where $p = y\te(y)$, or $(pqpx)^2 = qw^m$ (when $n = 4$) where $q = x\te(x)$, for some $x,y \in \Sigma^+$. 
	In both cases, if $|p| \le |q|$, Lemma~\ref{lem:pq-qp} can be applied and we have $\rho(p) = \rho(q)$, so $v = pq$ would not be even primitive. 
	Hence $|p| > |q|$ must hold, but then $|u| \ge 2|q|$ and then Proposition~\ref{prop:u2qwm3odd} implies $\rho(p) = \rho(q)$.

	Next we consider the case when $u_3 = \te(u)$ and $n = 3$.
	Then $\te(u) = \te(y)qp$ so that $\te(y)qp\te(y)qp = qw^m$. 
	Let $\te(y)q = qz$ for some $z$ with $|y| = |z|$. 
	Using $pq = y\te(y)q = yqz$, from $\te(y)qp\te(y)qp = qw^m$ we can obtain $zyqzzy\te(y) = w^m$. 
	Since $w = \te(w)$, this equation gives $z = y = \te(y)$. 
	Then $\te(y)q = qz$ turns into $yq = qy$ and hence $\rho(y) = \rho(q)$ by Theorem~\ref{th:uv-expr}. 
	This however implies that $v = y\te(y)q$ would not be $\te$-primitive.

	Finally we consider the case when $u_3 = \te(u)$ and $n = 4$.
	Then we have $[\te(x)pqp]^2 = qw^m$, which gives $x = \te(x)$ because $q = x\te(x)$. 
	Then $\te(u)^2 = x^2 w^m$, which is an LS equation and Theorem~\ref{thm:original} implies $\rho(\te(u)) = \rho(x) = w$. 
	However since $x^2p = qp \in \Suff(\te(u))$, we also get $\rho(p) = w$ (otherwise $w$ would be a proper infix of its square in $x^2$). This leads to the usual contradiction that $v = px^2$ would not be primitive. 
\end{proof}

\begin{proposition}
	Under the setting of Problem \ref{prob:main}, if $u_1u_2 = u\te(u)$, $u_3 = u_4$, and $m$ is even, then $u, v, w \in \{t, \te(t)\}^+$ for some $t \in \Sigma^+$. 
\end{proposition}
\begin{proof} 
	As before we consider only non-trivial equation so that we have $u_3 u_4 = q w_1 \cdots w_m$ and $|v| \neq |w|$. 
	Lemma~\ref{lem:u2qwm2} gives two cases, but actually it suffices to consider the case when $u = qr$ for some non-empty $\te$-palindrome $r$. 

	First we consider the case when $u_3 = u$ and $n$ is even. 
	Then $[(pq)^{n/2-1}px]^2 = q w_1 \cdots w_m$, where $q = x\te(x)$ for some $x \in \Sigma^+$. 
	If $|p| \le |q|$, then $pq = qp$ and $v$ would not be even primitive. 
	Hence let $p = qz_1$ for some $z_1 \in \Sigma^+$. 
	Then $r = z_1x\te(x)(pq)^{n/2-2}x\te(x)z_1x$. 
	Since $r = \te(r)$, this equation gives $z_1x = \te(z_1x)$ and $x = \te(x)$. 
	Thus we have $z_1x = x\te(z_1)$ and $p = x^2 z_1 = \te(z_1) x^2$. 
	Then $x^3 z_1 = x \te(z_1) x^2 = z_1 x^3$ so that $\rho(x) = \rho(z_1)$ by Theorem~\ref{th:uv-expr}.
	However, this result contradicts the primitivity of $v = pq = x^2z_1x^2$.

	The second case is when $u_3 = u$ an $n$ is odd. 
	We have $[(pq)^{(n-1)/2}y]^2 = q w_1 \cdots w_m$, where $p = y\te(y)$. 
	From this equation, $q$ is of even length so let $q = x\te(x)$. 
	If $|p| \le |q|$, then we can apply Lemma~\ref{lem:pq-qp} to the equation above to prove that $\rho(p) = \rho(q)$, which contradicts the primitivity of $v$.
	Thus we can let $y = xz_2$ for some $z_2 \in \Sigma^+$. 
	Then $[(xz_2\te(z_2)\te(x)x\te(x))^{(n-1)/2}xz_2]^2 = x\te(x) w_1 \cdots w_m$. 
	We can easily check that $w_{m/2+1} \cdots w_m = z_2 [\te(z_2)\te(x)x\te(x)xz_2]^{(n-1)/2}$. 
	According to Proposition~\ref{prop:uteuqnu}, we can deduce from this that $z_2, \te(x)x \in \{w, \te(w)\}^+$ and this further implies $x \in \{w, \te(w)\}^+$. 
	However then $v = pq = xz_2\te(z_2)\te(x)x\te(x)$ would not be $\te$-primitive. 

	Thirdly we consider the case when $u_3 = \te(u)$ and $n$ is even. 
	We have $[\te(x)p(qp)^{n/2-1}]^2 = x\te(x) w_1 \cdots w_m$, and this equation immediately gives $x = \te(x)$. 
	Then $p(qp)^{n/2-1}xp(qp)^{n/2-1} = x w_1 \cdots w_m$. 
	Since the left-hand side and $x$ are $\te$-palindromes, we have either $x \in \{w, \te(w)\}^+$ or $w_1 = \cdots = w_m = w$ by Lemma~\ref{lem:pali_pref_pali}. 
	In the former case, $\te(u)^2 = x^2 w_1 \cdots w_m \in \{w, \te(w)\}^+$ and hence $\te(u), u \in \{w, \te(w)\}^+$ (Lemma~\ref{lem:rt_th-rootshare}). 
	Then $v^n = u\te(u)x\te(x) \in \{w, \te(w)\}^+$, and hence $v \in \{w, \te(w)\}$ because of Lemma~\ref{lem:rt_th-rootshare} and the $\te$-primitivity of $v, w$. 
	However, this contradicts the assumption $|v| \neq |w|$. 
	In the latter case, we have $\te(u)^2 = x^2 w^m$ and hence $\rho(\te(u)) = \rho(x) = w$ (Theorem~\ref{thm:original}). 
	However since $qp = x^2p \in \Suff(\te(u))$, we reach the contradictory result $\rho(p) = w$. 

	The final case is when $u_3 = \te(u)$ and $n$ is odd. 
	Then $[\te(y)(qp)^{(n-1)/2}]^2 = q w_1 \cdots w_m$, where $p = y\te(y)$ for some $y \in \Sigma^+$. 
	Let $\te(y)q = qz_4$ for some $z_4$ with $|y| = |z_4|$. 
	Then $r = z_4(y\te(y)q)^{(n-1)/2}y\te(y)$, which is a $\te$-palindrome so that $z_4 = y = \te(y)$. 
	Now we can transform $\te(y)q = qz_4$ into $yq = qy$ and hence $\rho(y) = \rho(q)$ (Theorem~\ref{th:uv-expr}). 
	However, then $v = y\te(y)q$ would not be $\te$-primitive. 
\end{proof}

Combining the results obtained in this section, we can give a positive answer to Problem~\ref{prob:main}. 
Furthermore, with the result proved in~\cite{CCKS09} (also see Table~\ref{tbl:exLS_summary}), this positive answer concludes the following theorem, the strongest positive result we obtain on the ExLS equation. 

\begin{theorem}\label{thm:exLS4}
	Let $u, v, w \in \Sigma^+$ and let $u_1, \ldots, u_\ell \in \{u, \te(u)\}$, $v_1, \ldots, v_n \in \{v, \te(v)\}$, and $w_1, \ldots, w_m \in \{w, \te(w)\}$. 
	For $\ell \ge 4$ and $n, m \ge 3$, the equation $u_1 \cdots u_\ell = v_1 \cdots v_n w_1 \cdots w_m$ implies $u, v, w \in \{t, \te(t)\}^+$ for some $t \in \Sigma^+$. 
\end{theorem}

	\subsection{The case $\ell \le 3$ of the ExLS equation}\label{subsec:ExLS3}

We conclude this section with some examples which prove that an extended Lyndon-Sch\"{u}tzenberger theorem cannot be stated for $\ell = 2$, and for some particular cases when $\ell = 3$.

\begin{example}\label{ex:ExLS2}
	Let $\Sigma = \{a, b\}$ and $\theta$ be an antimorphic involutions on $\Sigma^*$ defined as $\theta(a) = a$ and $\theta(b) = b$. 
	Let $v = a^{2m} b^{2}$ and $w = aa$ (i.e., $w = \theta(w)$) for some $m \ge 1$. 
	Then $v^n w^m = (a^{2m} b^{2})^n a^{2m}$. 
	By letting either $u = (a^{2m} b^{2})^{n/2} a^m$ if $n$ is even or $u = (a^{2m} b^{2})^{(n-1)/2} a^{2m} b$ otherwise, we have $u\te(u) = v^n w^m$. 
	Nevertheless, there cannot exist a word $t$ such that $u, v, w \in \{t, \te(t)\}^+$ because $v$ contains $b$, while $w$ does not. 
	In conclusion, for arbitrary $n, m \ge 2$, $(2, n, m)$ does not impose $\te$-periodicity.
\end{example}

Next we examine briefly the $(3, n, m)$ ExLS equation. 
The actual problem which we address is formalized as follows: 

\begin{problem}\label{prob:ExLS3}
	Let $u, v, w \in \Sigma^+$ and integers $n, m \ge 3$. 
	Then, let $u_1, u_2, u_3 \in \{u, \te(u)\}$, $v_1, \ldots, v_n \in \{v, \te(v)\}$, and $w_1, \ldots, w_m \in \{w, \te(w)\}$. 
	Does the equation $u_1u_2u_3 = v_1 \cdots v_n w_1 \cdots w_m$ imply $u, v, w \in \{t, \te(t)\}^+$ for some $t \in \Sigma^+$ under all of the following conditions? 
	\begin{arabiclist}
	\item	$v$ and $w$ are $\te$-primitive, 
	\item	$|v_1 \cdots v_n| \ge |w_1 \cdots w_m|$, 
	\item	$u_1 = u$, $v_1 = v$, and $w_m = w$. 
	\end{arabiclist}
\end{problem}

As shown from now by examples, the general answer is ``No''. 
More significant is the fact that depending on the values of variables $u_2, u_3$ and on the lengths of $v_1 \cdots v_n$ and $w_1 \cdots w_m$, the $(3, n, m)$ ExLS equation exhibits very complicated behavior. 

First we present a parameterized example to show that for arbitrary $m \ge 2$, $(3, 3, m)$ does not impose $\te$-periodicity. 

\begin{example}\label{ex:ExLS33m}
	Let $\Sigma = \{a, b\}$ and $\theta$ be the mirror image over $\Sigma^*$. 
	For $u = (abb)^{2m-1}ab$, $v = (abb)^{m-1}ab$, and $w = (bba)^3$, we have $u^2 \te(u) = v \te(v)^2 w^m$ for any $m \ge 2$. 
	Nevertheless, there does not exist a word $t \in \Sigma^+$ satisfying $u, v, w \in \{t, \te(t)\}^+$. 
\end{example}

In this example, the border between $v\te(v)^2$ and $w^m$ is located at $u_2$. 
Intriguingly, as long as $u_1u_2u_3 = uu\te(u)$ we cannot shift the border to $u_3$ without imposing $u, v, w \in \{t, \te(t)\}^+$ for some $t \in \Sigma^+$. 

\begin{proposition}
	For any $n, m \ge 3$, if $uu\te(u) = v_1 \cdots v_n w_1 \cdots w_m$ and $n|v| > 2|u|$, then $u, v, w \in \{t, \te(t)\}^+$ for some $t \in \Sigma^+$. 
\end{proposition}
\begin{proof}
	It suffices to consider the case when $(n-1)|v| < 2|u| < n|v|$, otherwise Theorem~\ref{thm:exFWgcd} applies. 
	As done in the analyses on the ExLS equation with $\ell = 4$, we can assume that both $v$ and $w$ are $\te$-primitive. 
	Then, using Proposition~\ref{prop:CCKS09-a}, we obtain that $n$ is even, $u = [r(tr)^i r(tr)^i (tr)^j]^{n/2-1} r(tr)^i (rt)^j$ and $v = r(tr)^i r(tr)^i (tr)^j$ for some $i \ge 0$, $j \ge 1$, and two non-empty $\te$-palindromes $r, t$ such that $rt$ is primitive. 
	Moreover, $\te(u) = (tr)^j r(tr)^i [(rt)^j r(tr)^i r(tr)^i]^{n/2-1} = r(tr)^i r(tr)^i w_1 \cdots w_m$. 
	Hence if $i \ge 1$, then $tr = rt$, which contradicts the primitivity of $rt$ (Theorem~\ref{th:uv-expr}). 
	Thus we have 
	\begin{equation}\label{eq:exLS-001-3-1}
		(tr)^j r [(rt)^j r^2]^{n/2-1} = r^2 w_1 \cdots w_m.
	\end{equation}

	If $|t| \leq |r|$, then $t \in \Pref(r)$ from which $rt \in \Pref (r^2 w_1 \cdots w_m)$, and finally $rt = tr$, contradicting the primitivity of $rt$ again. 
	If $|r| < |t| \leq 2|r|$, then we can write $rrs = tr$ for some $s \in \Sigma^+$ such that $|r|+|s|=|t|$. 
	Since $s \in \Suff(r)$ and $r$ is a $\te$-palindrome, $\te(s) \in \Pref(r)$, i.e., $r=\te(s)r_1$ for some $r_1 \in \Sigma^+$. 
	Then, $rrs = r\te(s)r_1s = tr$, so $r\te(s) = t$ because their length is the same. 
	Since $\te(s) \in \Suff(t)$ and $t$ is a $\te$-palindrome, it holds that $s \in \Pref(t)$ and $rrs \in \Pref(rrt)$. 
	Therefore, $rrt$ and $tr$ share a prefix of length $|t|+|r|$ so that Theorem~\ref{th:uv-expr} concludes that $\rho(r) = \rho(t)$, contradicting the primitivity of $rt$.

	Thus both $i = 0$ and $|t| > 2|r|$ must hold. 
	Eq.~(\ref{eq:exLS-001-3-1}) implies that $r^2 \in \Pref(t)$, that is, $r^2 \in \Suff(t)$ ($t$ is a $\te$-palindrome), and hence $r^4 \in \Suff((rt)^j r^2)$. 
	So we can let $r^4 = z_1 w_{k+1} \cdots w_m$ for some $k \ge 1$ and $z_1 \in \Suff(w_k)$. 
	If $z_1 = \lambda$, then this equation gives $r \in \{w, \te(w)\}^+$ because $w$ is assumed to be $\te$-primitive due to Theorem \ref{thm:exFWlcm}. 
	Then Eq.~(\ref{eq:exLS-001-3-1}) means $(tr)^j r [(rt)^j r^2]^{n/2-1} \in \{w, \te(w)\}^+$. 
	Using Proposition~\ref{prop:pali_split}, we obtain $t \in \{w, \te(w)\}^+$, but this contradicts the $\te$-primitivity of $v$. 
	Otherwise, catenating $r^2$ from the left to the both sides of Eq.~(\ref{eq:exLS-001-3-1}) gives us $r[(rt)^j r^2]^{n/2} = z_1 w_{k+1} \cdots w_m w_1 \cdots w_m$. 
	Note that the left hand side of this equation is a $\te$-palindrome so that Lemma~\ref{lem:pal_x1-xn} implies $w_1 = \cdots = w_m = w$. 
	Now catenating $r$ in the same way to Eq.~(\ref{eq:exLS-001-3-1}) gives $[(rt)^j r^2]^{n/2} = r^3 w^m$. 
	This is in the form of LS equation and Theorem~\ref{thm:original} implies $\rho((rt)^j r^2) = \rho(r) = w$ because $w$ is primitive. 
	From this we further deduce that $\rho(t) = w$. 
	However, then $rt$ would not be primitive. 
\end{proof}

Once we change $u_1u_2u_3$ from $u^2\te(u)$ to $u\te(u)^2$, it becomes possible to construct a parameterized example for $(3, 3, m)$ with the border between $v_1 \cdots v_n$ and $w_1 \cdots w_m$ on $u_3$, though it works only when $m$ is a multiple of 3. 

\begin{example}
	Let $\Sigma = \{a, b\}$ and $\theta$ be the mirror image over $\Sigma^*$. 
	For $i, j \ge 0$, let $u = (ab)^{i+j+1} (ba)^{2i+2j+2} b(ab)^j$, $v = (ab)^{i+j+1} (ba)^{i+2j+1} b$, and $w = ab$. 
	Then $u\te(u)^2 = v^3 w^{2(i+j+1)} \te(w)^{i+j+1}$, but we cannot find such $t$ that $u, v, w \in \{t, \te(t)\}^+$. 
\end{example}

Next we increase $n$ to 4, and prove that still we can construct a parameterized example of the $(3, 4, 2i)$ ExLS equation. 

\begin{example}\label{ex:ExLS34even}
	Let $\Sigma = \{a, b\}$ and $\theta$ be the mirror image over $\Sigma$. 
	For $i \ge 1$, let $u = a^4(ba^3)^i(a^3b)^i$, $v = a^4(ba^3)^{i-1}ba^2$, and $w = ba^3$.
	Then we have $u^3 = v^2 \te(v)^2 w^i \te(w)^i$, but there does not exist a word $t \in \Sigma^+$ satisfying $u, v, w \in \{t, \te(t)\}^+$.  
\end{example}

The cases $(3,n,m)$ when $n=4$ and $m$ is odd, as well as when $m,n \geq 5$, remain open.

\begin{table}[h]
\tbl{Updated summary on the results regarding the extended Lyndon-Sch\"{u}tzenberger equation\label{tbl:new}}
{\begin{tabular}{r@{\hspace{8mm}}r@{\hspace{8mm}}r@{\hspace{10mm}}c@{\hspace{10mm}}l}
\toprule
$l$ & $n$ & $m$ & $\te$-periodicity & \\
\colrule
$\ge 4$ & $\ge 3$ & $\ge 3$ & YES & Theorem \ref{thm:exLS4} \\
\colrule
$3$ & $\ge 5$ & $\ge 5$ & ? & \\
$3$ & $4$ & odd & ? & \\
\colrule
$3$ & $4$ & even & NO & Example \ref{ex:ExLS34even} \\
$3$ & $3$ & $\ge 3$ & NO & Example \ref{ex:ExLS33m} \\
\multicolumn{3}{c}{one of them is 2} & NO & Example \ref{ex:ExLS2} \\
\botrule
\end{tabular}}
\end{table}

	\section{Conclusion}\label{sec:conclusion}

In this paper, we proved several consequences of the overlap between pseudo-primitive words. 
They made it possible to prove that, for a given antimorphic involution $\te$ and words $u, v, w \in \Sigma^+$, if $\ell \ge 4$ and $n, m \ge 3$, then the ExLS equation $u_1 \cdots u_\ell = v_1 \cdots v_n w_1 \cdots w_m$ implies that $u, v, w \in \{t, \te(t)\}^+$
 for some $t$. 
This is the strongest result obtained so far on the ExLS equation. 
Our case analyses on $(3, \ge 3, \ge 3)$ ExLS equations demonstrated that these tools may not be sufficient to provide a complete characterization of ExLS equations. 
Further investigation on the overlaps of $\te$-primitive words, reduction schemes from ExLS equations to LS equations, and the weak defect effect seems promising and required to fill the gap in Table \ref{tbl:new}.

\section*{Acknowledgments}

This research was supported by  Natural Sciences and Engineering Research Council of Canada Discovery Grant R2824A01,
 and Canada Research Chair Award to L.K.

	\bibliographystyle{plain}
	\bibliography{lyndonsch234}

\end{document}